\documentclass[journal]{IEEEtran}

\usepackage{cite}
\usepackage{amsmath,amssymb,amsfonts}
\usepackage{algorithm}
\usepackage{algpseudocode}
\usepackage{graphicx}
\usepackage{mathtools}
\usepackage{textcomp}
\usepackage{mathabx}
\usepackage{dblfloatfix}
\usepackage{subfig}
\usepackage{tabularx}
\usepackage{amsthm}
\def\BibTeX{{\rm B\kern-.05em{\sc i\kern-.025em b}\kern-.08em
    T\kern-.1667em\lower.7ex\hbox{E}\kern-.125emX}}

\makeatletter
\newcolumntype{Y}{>{\centering\arraybackslash}X}

\newcommand{\average}[2][0]{{%
		\mspace{#1mu}%
		\overline{\mspace{-#1mu}\average@check#2\relax}%
}}
\newcommand\average@check[1]{%
	#1\@ifnextchar_{\average@sub}{}%
}
\newcommand{\average@sub}[2]{
	_{#2}\mspace{-2mu}\aftergroup\average@compensate
}
\newcommand{\average@compensate}{\mspace{2mu}}

\newcommand{\baverage}[2][0]{{%
		\mspace{#1mu}%
		\underline{\mspace{-#1mu}\baverage@check#2\relax}%
}}
\newcommand\baverage@check[1]{%
	#1\@ifnextchar_{\baverage@sub}{}%
}
\newcommand{\baverage@sub}[2]{
	_{#2}\mspace{-2mu}\aftergroup\baverage@compensate
}
\newcommand{\baverage@compensate}{\mspace{2mu}}
\makeatother
\newtheorem{theorem}{Theorem}
\newtheorem{lemma}{Lemma}
\newtheorem{assumption}{Assumption}
\newtheorem{remark}{Remark}

\newtheorem{proposition}{Proposition}

\newcommand{\norm}[1]{\left\lVert#1\right\rVert}

\usepackage{eso-pic}

\begin{document}
	\title{A Novel Approach to Set-Membership Observer for Systems with Unknown Exogenous Inputs}
	\author{Marvin Jesse, Dawei Sun, \IEEEmembership{Student Member, IEEE}, and Inseok Hwang, \IEEEmembership{Member, IEEE}
		\thanks{This work is supported in part by NSF CNS-1836952. }
		\thanks{The authors are with School of Aeronautics and Astronautics, Purdue University, West Lafayette, IN 47907 USA (e-mail: jessem@purdue.edu, sun289@purdue.edu, ihwang@purdue.edu).}
	}
\maketitle

	\begin{abstract}
		Motivated by the increasing need to monitor safety-critical systems subject to uncertainties, a novel set-membership approach is proposed to estimate the state of a dynamical system with unknown-but-bounded exogenous inputs. 
		The proposed method decomposes the system into the strongly observable and weakly unobservable subsystem in which an unknown input observer and an ellipsoidal set-membership observer are designed for each subsystem, respectively.
		The conditions for the boundedness of the proposed set estimate are discussed, and the proposed set-membership observer is also tested numerically using illustrative examples.
	\end{abstract}
	
	
	\section{Introduction}
	\label{sec:introduction}
	
State estimation has been widely used in the control community in areas like the secure control of cyber-physical systems \cite{c34} and fault diagnosis \cite{c35}. One of the common methods of state estimation used is the deterministic method, which treats the noise and disturbance as unknown-but-bounded (UBB) \cite{c38}.

\textit{Related works: }One of the most well-established methods to estimate the system's state under UBB uncertainty is the set-membership observer. This method uses geometrical sets, such as ellipsoids \cite{c21}, zonotopes \cite{c3}, or parallelotopes \cite{c52}, to enclose all admissible state values. 
An alternative approach is the interval observer. This method works by evaluating the error dynamics generated by the upper and lower bounds of the estimated states so that the error dynamics are cooperative and stable \cite{c5}. The unknown input observer \cite{c14} is a third viable option that can accurately estimate the system state without much prior knowledge of the inputs, rendering it more amenable for use in fault detection schemes. The design of such an observer does require the system to have strong observability, which limits its applicability. Strong observability \cite{c31} is defined as a system's property in which we can infer the true state from the system's output for any initial state and unknown input.

Note that each class of approaches has its own advantages and disadvantages. In particular, the interval observer offers low computational complexity but more conservative results than the set-membership observer, and vice versa. In order to achieve a balance between these two specific aspects, the authors in \cite{c11} combined both the interval observer and set-membership observer. Additionally, the idea to decompose the system into strongly observable and weakly unobservable subsystems has been attempted in \cite{c53}. In \cite{c53}, the authors apply a High Order Sliding Mode technique and interval observer in order to improve the estimation accuracy. 
Motivated by these works, we consider a more effective way to use the structure of the system such that our proposed set-membership observer has a comparable estimation performance while not requiring the system to be too restrictive. 

	\begin{figure*}[b]
	\centering
	\includegraphics[clip,trim=0.001cm 12.8cm 0.5cm 8.2cm, width=0.98\textwidth]{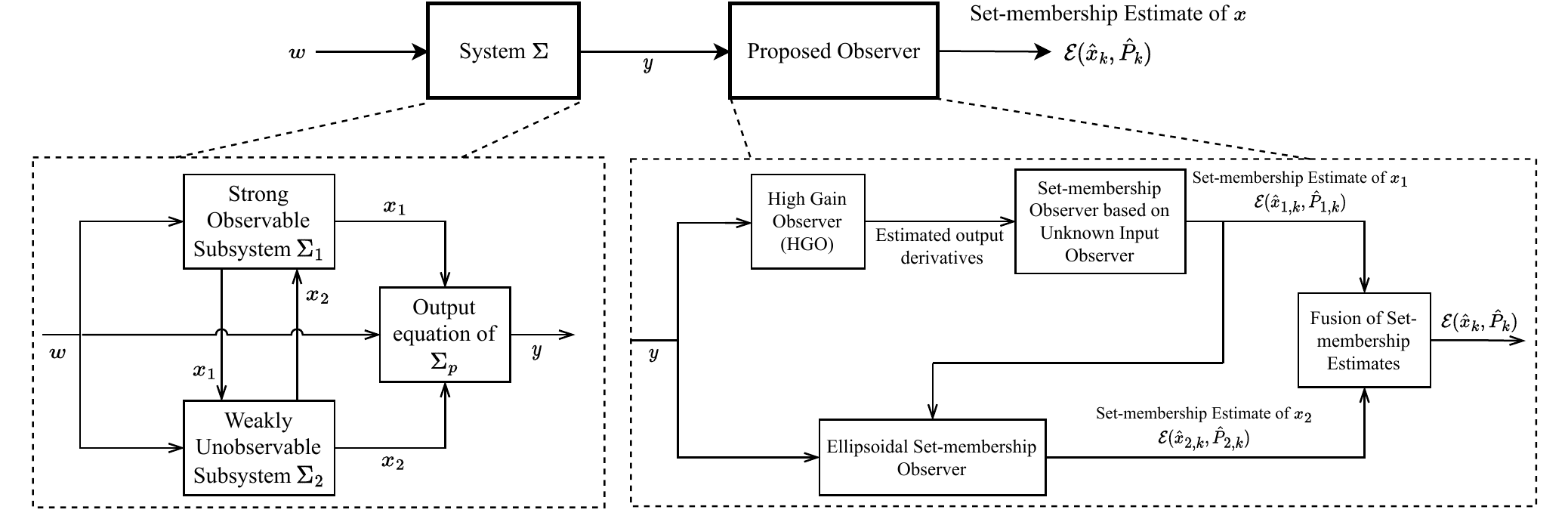}
	\caption{Proposed observer architecture for system $\Sigma$}
	\label{fig1}
\end{figure*}

\textit{Contributions:}	This paper presents a novel set-membership approach to state estimation of the linear time-invariant (LTI) system that integrates the unknown input observer and the ellipsoidal set-membership observer, which has not been attempted in literature, to the best of our knowledge. 
Based on a system decomposition technique, we implement an unknown input observer and an ellipsoidal set-membership observer for the strongly observable and weakly unobservable subsystem, respectively.
Compared with existing representative works, our contributions include relaxing one of the assumptions in \cite{c53}, which is the existence of a transformation matrix to transform the weakly unobservable subsystem into a cooperative form, and providing analytical analysis in which the set estimate computed by our observer is stable, i.e., the set estimate does not grow to infinity, for which less restrictive conditions are required compared to \cite{c19,c20,c21}.
In terms of estimation accuracy, our proposed observer outperforms existing ellipsoidal set-membership observers as well, which is demonstrated thoroughly in the numerical simulations.

	The rest of this paper is organized as follows: we first introduce the preliminaries and problem setup in Section II. The detailed design of the observer is described in Section III. Section IV discusses properties of the proposed algorithm. Two numerical simulations are presented in Section V to illustrate the effectiveness of our proposed algorithm. Finally, Section VI concludes the paper. 
	
	\section{Preliminaries and Problem Setup}
Throughout this paper, we denote an ellipsoid as $\mathcal{E}\big(c,K\big) \triangleq \{x \in R^{x}: (x-c)^TK^{-1}(x-c)\leq 1\}$, where $c\in \mathbb{R}^{x}$  is the center vector and $K \in \mathbb{R}^{x\times x}$ is a symmetric positive definite matrix, called the shape matrix. The pseudoinverse and transpose of a matrix $A$ are denoted as $A^\dagger$ and $A^T$, respectively. $\mathrm{rank}(A)$ and $\mathrm{tr}(A)$ denotes the rank and trace of $A$, respectively. $\mathrm{col}(a_1,a_2,\hdots,a_n)$ denotes a column vector. $||\cdot||$ denotes the standard 2-norm. $\mathrm{diag}(A_1,A_2,\hdots,A_n)$ denotes a block diagonal matrix in which the diagonal elements are $A_1$, $A_2$, $\hdots$, $A_n$.
		$\max Re(\lambda(A))$ denotes the largest real part of $A$'s eigenvalues. $\sigma_{min}(A)$ denotes the smallest singular value of $A$.
	
	In this paper, we will consider the following continuous linear time-invariant (LTI) system
	\begin{equation} \label{eq:1}
		\Sigma: 
		\left\{
		\begin{aligned}
			\dot{x}(t) &= Ax(t) + Bw(t)\\
			y(t) &= Cx(t) + Dw(t)       
		\end{aligned}
		\right.,
	\end{equation}
	where $A \in \mathbb{R}^{n\times n}$, $B \in \mathbb{R}^{n\times n_w}$, $C \in \mathbb{R}^{n_y \times n}$, and $D \in \mathbb{R}^{n_y\times n_w}$ are constant matrices; $x(t)\in \mathbb{R}^{n}$ is the unknown state to be estimated; $y(t) \in \mathbb{R}^{n_y}$ is the measurable output; and $w(t) \in \mathbb{R}^{n_w}$ represents the unknown input vector. The following assumptions are commonly made for set-membership state estimation problems \cite{c19,c20,c21}.
	\begin{assumption} \label{ass:1}
		The unknown input $w(t)$ is bounded, i.e., $w(t) \in \mathcal{E}\big(c_w(t),K_w(t)\big)$.
	\end{assumption}
	\begin{assumption} \label{ass:2}
		The initial state $x(0)$ is bounded, i.e., $x(0) \in \mathcal{E}\big(\hat{x}_0,K_0\big)$.
	\end{assumption}
Given the system $\Sigma$ in (\ref{eq:1}) with its parametric matrices, the output $y(t)$, bound on the unknown input $\mathcal{E}(c_w(t),K_w(t))$, and bound on the initial state $\mathcal{E}(\hat{x}_0,K_0)$, our objectives are to
	design a computationally efficient and accurate state estimation algorithm such that $x(t_k) \in \mathcal{E}(\hat{x}_k,\hat{P}_k)$ for all $k\in\mathbb{N}$ and
	to investigate the conditions for which $\mathcal{E}(\hat{x}_k,\hat{P}_k)$'s are uniformly bounded.

	\section{Observer Design}

	In this section, we first discuss our proposed observer architecture and system decomposition. Then, we proceed by designing an observer for each subsystem and fusing these individual set estimates into a single set estimate. Finally, the proposed approach is summarized in Algorithm 1.
	\subsection{Proposed Observer Architecture and System Decomposition}

	To solve the aforementioned problem, we are motivated to utilize the system's structure, which can be done by decomposing the given system $\Sigma$ into two subsystems: a strongly observable subsystem and a weakly unobservable subsystem such that a different observer can be designed to match each subsystem.
	Our proposed scheme, which can facilitate the characteristics of each subsystem, is illustrated in Figure \ref{fig1}.

   To achieve the desired system decomposition, we follow the methods described in \cite{c15,c16} to obtain the transformation matrix. 
	Next, the following lemma illustrates how to decompose the system $\Sigma$.
	\begin{lemma}[System Decomposition] \label{lem:1}
		For the system $\Sigma$ in (\ref{eq:1}), there exists a coordinate transformation matrix $P_1$ such that for $x_p=\mathrm{col}(x_1,x_2)= P_1x$, one has
		\begin{equation} \label{eqn:1}
			\Sigma_p: 
			\left\{
			\begin{aligned}
				\dot{x}_p(t) &= \underbrace{\begin{bmatrix}A_1&A_3\\A_2&A_4\end{bmatrix}}_{A_p}x_p(t) + \underbrace{\begin{bmatrix}B_1\\B_2\end{bmatrix}}_{B_p}w(t)\\
				y(t) &= \underbrace{\begin{bmatrix}C_1&C_2\end{bmatrix}}_{C_p}x_p(t) + Dw(t)  
			\end{aligned}
			\right.,
		\end{equation}
		where $A_p = P_1AP_1^{-1}$, $B_p = P_1B$, $C_p = CP_1^{-1}$, $x_1 \in \mathbb{R}^{n_1}$, and $x_2 \in \mathbb{R}^{n_2}$. Additionally, the system $\Sigma_p$ satisfies the property that subsystems $\Sigma_1$ and $\Sigma_2$, which are defined as
		\begin{equation}\label{eq:3}
			\Sigma_1:
			\left\{
			\begin{aligned}
				\dot{x}_1&=A_1x_1+B_1'u_1\\
				y &= C_1x_1+D_1'u_1
			\end{aligned}
			\right., 
			\Sigma_2:
			\left\{
			\begin{aligned}
				\dot{x}_2 &= A_4x_2+B_2'u_2\\
				y&=C_2x_2+D_2'u_2
			\end{aligned}
			\right.,
		\end{equation}
		where $
		B_1' = \begin{bmatrix}A_3&B_1\end{bmatrix}$, $D_1' = \begin{bmatrix}C_2&D\end{bmatrix}$, and $u_1 = \mathrm{col}(x_2,w)$,
	   $
		B_2' = \begin{bmatrix}
			A_2&B_2
		\end{bmatrix}$, $D_2' = \begin{bmatrix}
			C_1&D
		\end{bmatrix}$, and
		$u_2 = \mathrm{col}(x_1,w)$,
		are strongly observable and weakly unobservable \cite{c15}, respectively.
	\end{lemma}
	\begin{proof}
		See Lemma 3.4 in \cite{c15} and Theorem 4 in \cite{c16}.
	\end{proof}

	\subsection{Set-Membership Observer for $\Sigma_1$}
	A set-membership observer based on the unknown input observer is proposed for subsystem $\Sigma_1$, and the derivation is adapted from \cite{c12} to be self-contained.
	The goal is to find a bounding ellipsoid $\mathcal{E}(\hat{x}_1(t),\epsilon_{1}^2(t)I_{n_1})$ that contains the true state $x_1(t)$. First, we are interested in finding the center of the ellipsoid $\hat{x}_1(t)$, which is obtained using an unknown input observer.
	The following assumption is needed for the observer:
	\begin{assumption} \label{ass:10}
		The unknown input $w(t)$ is a sufficiently smooth function, i.e., $w(t) \in \mathbb{C}^l$ for some $l$, and the derivatives of $w(t)$ are bounded. 
	\end{assumption}
	Under Assumption \ref{ass:10}, we can differentiate the output equation in (\ref{eq:3}) $l$ times to get
	\begin{equation} \label{eq:8}
		z_{0:l} = \mathcal{O}_{l}x_1+\mathcal{G}_{l}u_{1,0:l-1},
	\end{equation}
	where
	\begin{gather*}
		\mathcal{O}_{l} = \begin{bsmallmatrix}C_1\\C_1A_1\\\vdots\\C_1A_1^{l}\end{bsmallmatrix}, 
		\mathcal{G}_{l} = \begin{bsmallmatrix}D_1'&0&\hdots&0\\C_1B_1'&D_1'&\hdots&0\\\vdots&\vdots&\ddots&\vdots\\C_1A_1^{l-1}B_1'&C_1A_1^{l-2}B_1'&\hdots&D_1'\end{bsmallmatrix}, \\
		u_{1,0:l-1} = \mathrm{col}(u_1,\dot{u}_1,\hdots,u_1^{(l-1)}),
		z_{0:l}=\mathrm{col}(y,\dot{y},\hdots,y^{(l)}).
	\end{gather*}
	Then, the unknown input observer is given as
	\begin{equation} \label{eq:10}
		\dot{\hat{x}}_1(t) = E\hat{x}_1(t) + F\hat{z}_{0:l},
	\end{equation}
	where $F\in\mathbb{R}^{n_1\times (l+1)n_y}$ is a design matrix satisfying $F\mathcal{G}_l = \begin{bmatrix}B_1'&0&\hdots&0\end{bmatrix}$ and $E = A_1-F\mathcal{O}_l \in \mathbb{R}^{n_1\times n_1}$ is designed to be a stable matrix. Besides that, according to \cite{c31} and strong observability of subsystem $\Sigma_1$, there exists a smallest integer $l \leq n_1$ such that it satisfies $\mathrm{rank}(G_l)-\mathrm{rank}(G_{l-1})=n_2+w$. Here, $\hat{z}_{0:l}$ is a vector of estimated output derivatives, $\hat{z}_{0:l}$, i.e. $\hat{z}_{0:l} = \mathrm{col}(\hat{z}'_0,\hat{z}'_1,\hdots,\hat{z}'_{l})$ where $\hat{z}'_i=\mathrm{col}(\hat{z}_{1,i},\hat{z}_{2,i},\hdots,\hat{z}_{n_y,i})$. We will now discuss how to obtain $\hat{z}_{i,j}$, which is the main idea of the construction of a high gain observer (HGO). 
	
	For all $i=1,2,\hdots,n_y$ and $j=0,1,\hdots,l$, let $z_{i,j}=y^{(j)}_{i}$, where $y^{(j)}_{i}$ denotes the $j$-th derivative of the $i$-th element of $y$. Therefore, defining $z_i = \mathrm{col}(z_{i,0},z_{i,1},\hdots,z_{i,l})$ for $1 \leq i \leq n_y$,
	it can be approximately obtained by using a high gain observer \cite{c30} with the form 
	\begin{equation} \label{eq:6}
		\dot{\hat{z}}_i = \hat{A}_{z,i}\hat{z}_i+\hat{B}_{z,i}y_i,\\
	\end{equation}
	where
	\begin{gather*}
		\hat{A}_{z,i} = \begin{bmatrix}
			-\frac{\theta_{i,0}}{\epsilon}& 1 & 0 &\hdots & 0\\
			-\frac{\theta_{i,1}}{\epsilon^2} & 0 & 1&\hdots & 0\\
			\vdots & \vdots&\vdots&\ddots&\vdots\\
			-\frac{\theta_{i,l}}{\epsilon^{l+1}}&0&0&\hdots&0
		\end{bmatrix},
		\hat{B}_{z,i} = \begin{bmatrix}
			\frac{\theta_{i,0}}{\epsilon}\\\frac{\theta_{i,1}}{\epsilon^2}\\ \vdots \\ \frac{\theta_{i,l+1}}{\epsilon^{l+1}},
		\end{bmatrix}
	\end{gather*}
	for $1\leq i \leq n_y$, $\epsilon$ is a sufficiently small positive constant, and the coefficients $\theta_{i,j}$ are chosen such that the polynomial
	$
		s^{n_1+1}+\theta_{i,0}s^{n_1} + \hdots + \theta_{i,n_1}s+\theta_{i,l}
	$
	is Hurwitz for $i=1,2,\hdots,n_y$. 

	Then, the next step is to find an explicit expression of $\epsilon_{1}(t)$, which represents a time-varying bound on the estimation error between the true state $x_1(t)$ and its estimate $\hat{x}_1(t)$. Without loss of generality, suppose that the bound on the initial estimation error of the output derivatives is known, i.e., we have $\average{z}_0 \geq 0$ such that $||\tilde{z}_{i}(0)|| \leq \average{z}_0, \forall i \in [1,n_y]$. Next, let us define $\epsilon_1(t)$ as
	\begin{equation} \label{eq:12}
	\epsilon_1(t) =\norm{e^{Et}}\norm{P_1K_0P_1^T}^{\frac{1}{2}} + \norm{F}\left(n_y(l+1)\right)^{\frac{1}{2}} \Psi(t),
	\end{equation}
	where
	\begin{equation}\label{eq:DefPsi}
	\begin{aligned}
		\Psi(t) &= \delta\int_0^t \norm{e^{E(t-\tau)}} d\tau \\ &+\left(\frac{K\sqrt{l+1}}{\epsilon^l}\overline{z}_0-\epsilon^{l}\delta\right)\int_0^t \norm{e^{E(t-\tau)}}e^{-\frac{a\tau}{\epsilon}} d\tau.
	\end{aligned}
	\end{equation}
	$K$ and $a$ are chosen such that $||e^{A_{\eta_i}t}|| \leq Ke^{-at}$ is true for all $t$ and $i \in [1,n_y]$, where
		\begin{equation*}
			\small
		A_{\eta_i}= \begin{bmatrix}
			-\theta_{i,0}& 1 & 0 &\hdots & 0\\
			-\theta_{i,1} & 0 & 1&\hdots & 0\\
			\vdots & \vdots&\vdots&\ddots&\vdots\\
			-\theta_{i,l}&0&0&\hdots&0
		\end{bmatrix},
	\end{equation*}
	and $\delta$ should satisfy the inequality $\sup_{t\in[0,\infty)}\frac{||y^{(l)}||K\epsilon}{a} \leq \delta$.
	Finally, the following proposition illustrates the expression of the set-membership observer for subsystem $\Sigma_1$. 
	\begin{proposition} [Set-Membership Observer for $\Sigma_1$]\label{prop:1}
		Consider the strongly observable subsystem $\Sigma_1$ in (\ref{eq:3}) of the given system $\Sigma$ in (\ref{eq:1}) with Assumptions \ref{ass:2} and \ref{ass:10}.
		The set-membership observer for subsystem $\Sigma_1$ takes the form of $\mathcal{E}(\hat{x}_1(t),\epsilon_{1}^2(t)I_{n_1})$, where $\hat{x}_1(t)$ is the solution to (\ref{eq:10}), and $\epsilon_{1}(t)$ is given in (\ref{eq:12}),
		and the ellipsoid $\mathcal{E}(\hat{x}_1(t),\epsilon_{1}^2(t)I_{n_1})$ is guaranteed to contain the true state $x_1(t)$, i.e.,
		$
			x_1(t)\in \mathcal{E}(\hat{x}_1(t),\epsilon_{1}^2(t)I_{n_1})
		$
		 for all $t \geq 0$.
	\end{proposition}
	\begin{proof}
		See Appendix \ref{app:5} for the detailed proof.
	\end{proof}
	
	\subsection{Ellipsoidal Set-Membership Observer for $\Sigma_2$}
	In this section, we propose an ellipsoidal set-membership observer which is designed for the weakly unobservable subsystem $\Sigma_2$.
	Note that our set-membership observer will be implemented in a discrete-time manner, and we denote the discrete-time instances as $t_k, k=0,1,2,\hdots$ and $t_0=0$ without loss of generality.

	Suppose $x_2(t_k)$ is the true state of $x_2$ at time $t_k$, and our goal is to find a bounding ellipsoid $\mathcal{E}(\hat{x}_{2,k},\hat{P}_{2,k})$ that contains $x_2(t_k)$. Let $\mathcal{E}(\hat{x}_{2,k|k-1},\hat{P}_{2,k|k-1})$ be the propagated ellipsoid, which is the set of possible values of $x_2(t_{k})$ satisfying subsystem $\Sigma_2$'s dynamics. Let $S_k = \{x_2(t_k)\in \mathbb{R}^{n_2}: y(t_k)=C_2x_2(t_k)+D_2'u_2\}$ be the measurement set at time $t_k$. This set is intersected with the propagated ellipsoid and overapproximated, i.e. $\mathcal{E}(\hat{x}_{2,k|k-1},\hat{P}_{2,k|k-1}) \cap S_k \subset \mathcal{E}(\hat{x}_{2,k},\hat{P}_{2,k})$.
	Figure \ref{fig2} illustrates how the ellipsoidal set-membership observer is constructed. Besides that, without loss of generality, the matrices $B_2'$ and $D_2'$ have full row rank (we can introduce some transformation matrix if it is not satisfied).
	Now, we are ready to present the ellipsoidal set-membership observer. First, the propagation step takes the form
	\begin{align}
		\hat{x}_{2,k|k-1} &= e^{A_4\Delta t}\hat{x}_{2,k-1}+\int_{t_{k-1}}^{t_{k}}e^{A_4(t_{k}-\tau)}B_2'\begin{bmatrix}\hat{x}_1(\tau)\\c_w(\tau)\end{bmatrix}d\tau, \label{eq:50}\\
		\hat{P}_{2,k|k-1} &= \frac{e^{A_4\Delta t} \hat{P}_{2,k-1}e^{A_4^T\Delta t}}{\alpha_k} + \frac{\Delta tM_{2,k}}{1-\alpha_k} \label{eq:51},
	\end{align}
	and the measurement step is given as
	\begin{align}
		\hat{x}_{2,k} &= \hat{x}_{2,k|k-1} + O_k\left(y(t_k)-C_2\hat{x}_{2,k|k-1}-D_2'\begin{bmatrix}
			\hat{x}_1(t_k)\\c_w(t_k)
		\end{bmatrix}\right), \label{eq:22}\\
		\hat{P}_{2,k} &= \frac{1}{1-\beta_k}(I-O_kC_2)\hat{P}_{2,k|k-1}, \label{eq:23}
	\end{align}
	where
	\begin{gather*}
		M_{2,k} = \int_{t_{k-1}}^{t_k} e^{A_4(t_k-\tau)}B_2'K_u(\tau)B_2^{\prime T}e^{A_4^T(t_k-\tau)}d\tau,\\
		O_k = \frac{1}{1-\beta_k}\hat{P}_{2,k|k-1}C_2^T\left(\frac{1}{1-\beta_k}C_2\hat{P}_{2,k|k-1}C_2^T+\frac{G_k}{\beta_k}\right)^{-1},\\
		G_k = D_2'K_u(t_k)D_2^{\prime T}, K_u(t) = \begin{bmatrix}
			\gamma_k\epsilon_1^2(t)I_{n_1}&0\\0&\frac{\gamma_k}{\gamma_k-1}K_w(t)
		\end{bmatrix},
	\end{gather*}
	 and $\Delta t = t_k-t_{k-1}$. The following proposition shows the set-membership observer for subsystem $\Sigma_2$.
	\begin{proposition} [Set-Membership Observer for $\Sigma_2$] \label{prop:2}
		Consider the weakly unobservable subsystem $\Sigma_2$ in (\ref{eq:3}) of the given system $\Sigma$ in (\ref{eq:1}) with Assumptions \ref{ass:1} and \ref{ass:2}. If $x_2(t_{k-1}) \in \mathcal{E}(\hat{x}_{2,k-1},\hat{P}_{2,k-1})$, and we apply the propagation and measurement steps according to (\ref{eq:50}), (\ref{eq:51}), (\ref{eq:22}), and (\ref{eq:23}), and we choose $\alpha_k$, $\beta_k$ and $\gamma_k$ such that $\alpha_k \in (0,1)$, $\beta_k \in (0,1)$ and $\gamma_k>1$, the set-membership observer for subsystem $\Sigma_2$ takes the form of $\mathcal{E}(\hat{x}_{2,k},\hat{P}_{2,k})$, and the ellipsoid $\mathcal{E}(\hat{x}_{2,k},\hat{P}_{2,k})$ is guaranteed to contain the true state $x_2(t_k)$, i.e.,
		$
			x_2(t_{k}) \in \mathcal{E}(\hat{x}_{2,k},\hat{P}_{2,k})
		$
		for all $k\in \mathbb{N}$.
	\end{proposition}
	\begin{proof}
		The proof is adapted from Theorem 2 in \cite{c32}.
	\end{proof}
	\begin{figure}[!t]
		\centering
		\includegraphics[clip,trim=1cm 10.5cm 1cm 7.8cm, width=\columnwidth]{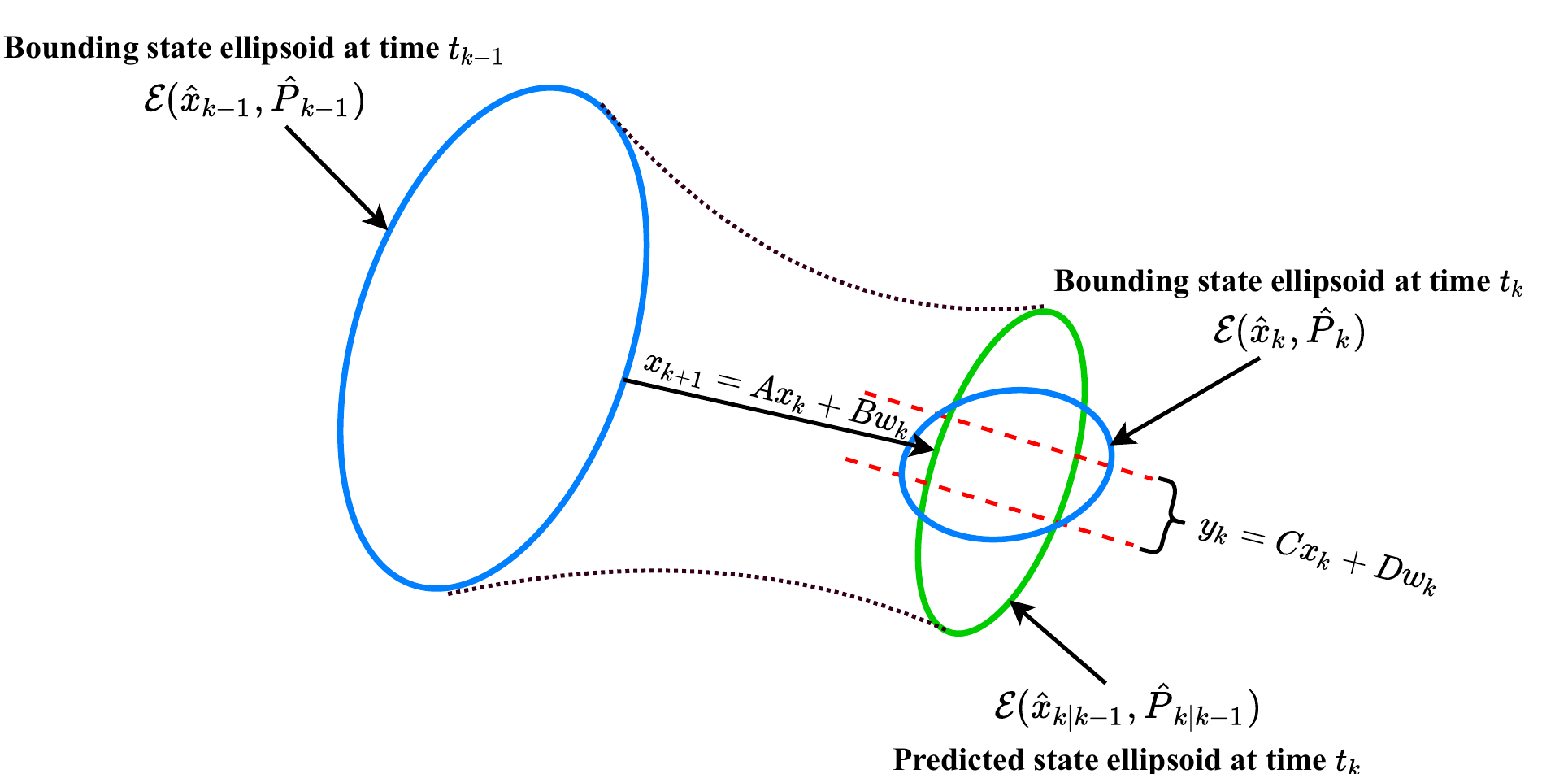}
		\caption{Propagation and measurement steps of the ellipsoidal set-membership observer}
		\label{fig2}
	\end{figure}

	Note that $\alpha_k$, $\beta_k$, and $\gamma_k$ are design parameters, and
	there are different criteria that we can consider, thus resulting in the most nonconservative bounding ellipsoid. As for finding $\alpha_k$, we can refer to the approach in \cite{c32}, which will minimize the trace of $\hat{P}_{2,k|k-1}$, and it yields the solution
	\begin{equation} \label{eq:91}
		\alpha_k = \frac{\sqrt{\mathrm{tr}(M_{2,k})}}{\sqrt{\mathrm{tr}(M_{2,k})}+\sqrt{\mathrm{tr}(e^{A_4\Delta t}\hat{P}_{2,k-1}e^{A_4^T\Delta t})}}.
	\end{equation}
 Since the observation step involves intersection of two ellipsoids, which is, in general, not an ellipsoid, it is more challenging to find the parameter $\beta_k$. Nevertheless, we can find $\beta_k$ by minimizing the trace of $\hat{P}_{2,k}$ through an optimization problem. The constraints of the optimization problem is obtained by rewriting (\ref{eq:23}), and it is as follows:
	\begin{equation} \label{eq:17}
	\begin{aligned}
		&\underset{\beta_k}{\text{minimize}}
		\qquad \mathrm{tr}\left(X_k^{-1}\right) \\
		& \text{subject to} \quad
		X_k = (1-\beta_k)\hat{P}_{2,k|k-1}^{-1}+\beta_kC_2^TG_k^{-1}C_2,\\
		& \qquad \qquad \qquad \qquad \quad 0< \beta_k <1, \hspace{0.25em} X_k \succ 0 
		\end{aligned}
	\end{equation}

	This optimization problem is convex and solvable using semidefinite programming (SDP). Lastly, the parameter $\gamma_k$ is associated with finding the minimal bounding ellipsoid for $u_2$, which involves an overapproximation of the Cartesian product of $\mathcal{E}(\hat{x}_1(t),\epsilon_1^2(t)I_{n_1})$ and $\mathcal{E}(c_w(t),K_w(t))$. The following lemma helps determine $\gamma_k$.
	\begin{lemma} \label{lem:4}
		Given $x_{q1} \in \mathcal{E}(\hat{x}_{q_1},Q_1)$ and $x_{q_2} \in \mathcal{E}(\hat{x}_{q_2},Q_2)$, where $\hat{x}_{q_1}\in\mathbb{R}^{n_{q_1}}$, $Q_1 \in \mathbb{R}^{n_{q_1}\times n_{q_1}}$, $\hat{x}_{q_2}\in\mathbb{R}^{n_{q_2}}$, and $Q_2 \in \mathbb{R}^{n_{q_2}\times n_{q_2}}$, we have $\mathrm{col}(x_{q_1},x_{q_2}) \in \mathcal{E}(\hat{x}_q,Q)$, where
		$
			\hat{x}_q = \mathrm{col}(\hat{x}_{q_1}, \hat{x}_{q_2})$ and
			$Q = \mathrm{diag}(
				gQ_1,\frac{g}{g-1}Q_2
			)$
		for all $g > 1$. 
		In addition, if 
		$
			g = \sqrt{\frac{\mathrm{tr}(Q_2)}{\mathrm{tr}(Q_1)}}+1,
		$
		then, the trace of $Q$ is minimized.
	\end{lemma}
	\begin{proof}
		See Appendix \ref{app:2}.
	\end{proof}
	Therefore, applying Lemma \ref{lem:4} to $K_u(t)$, $\gamma_k$, which minimizes $\mathrm{tr}(K_u(t))$, can be computed as
	\begin{equation} \label{eq:92}
		\small 
		\gamma_k = \sqrt{\frac{\mathrm{tr}(K_w(t_k))}{\mathrm{tr}(\epsilon_{1,k}^2I_{n_1})}}+1,
	\end{equation}
	where $\epsilon_{1,k}$ is the solution of (\ref{eq:12}) at time $t_k$.
	
	\begin{remark}
		In fact, there can be alternative criteria in choosing the parameters $\alpha_k$, $\beta_k$, and $\gamma_k$, which include volume minimization \cite{c9} and maximizing the decrease of a Lyapunov function on the estimation error \cite{c20}. For our proposed algorithm, we consider the minimum trace as a design criterion, and the readers might consider other criteria, depending on the specific applications.
	\end{remark}
	
	\subsection{Fusion of Set-Membership Estimates}
	In this section, we develop a method to combine the two set estimates from subsystems $\Sigma_1$ and $\Sigma_2$ into a set estimate for the original system $\Sigma$, i.e., we want to find a bounding ellipsoid $\mathcal{E}(\hat{x}_k,\hat{P}_k)$ that contains $x(t_k)$. The following theorem will illustrate how to integrate the set estimates from the two observers for subsystems $\Sigma_1$ and $\Sigma_2$.
	\begin{theorem}[Estimation Set of the Original State] \label{th:2}
		For the system $\Sigma$ in (\ref{eq:1}) with its transformed system $\Sigma_p$ in (\ref{eqn:1}), suppose $x_1(t_k) \in \mathcal{E}(\hat{x}_{1,k},\epsilon_{1,k}^2 I_{n_1})$ and $x_2(t_k) \in \mathcal{E}(\hat{x}_{2,k},\hat{P}_{2,k})$. Then, $x(t_k) \in \mathcal{E}(\hat{x}_k,\hat{P}_k)$, where
		$
			\hat{x}_k = P_1^{-1}\mathrm{col}(
				\hat{x}_{1,k},\hat{x}_{2,k}
			)$ and $ \hat{P}_k = P_1^{-1}\mathrm{diag}(
				\mu_k\epsilon_{1,k}^2I_{n_1}, \frac{\mu_k}{\mu_k-1} \hat{P}_{2,k}
			)P_1^{-T},
		$
		for all $\mu_k > 1$. In addition, if $\mu_k = \sqrt{\frac{\mathrm{tr}(\hat{P}_{2,k})}{\mathrm{tr}(\epsilon_{1,k}^2I_{n_1})}}+1$, $\mathrm{tr}(\hat{P}_k)$ is minimized.
	\end{theorem}
	\begin{proof}
		First, since $x(t_k) = P_1^{-1}\mathrm{col}(x_1(t_k),x_2(t_k))$ by Lemma \ref{lem:1}, one can rewrite $(x(t_k)-\hat{x}_k)^T\hat{P}_k^{-1}(x(t_k)-\hat{x}_k)$ as 
		\begin{align}\label{eq:24}
			&\frac{1}{\mu_k}(x_1(t_k)-\hat{x}_{1,k})^T\frac{1}{\epsilon_{1,k}^2}I_{n_1}(x_1(t_k)-\hat{x}_{1,k})\nonumber\\
			&\quad +\frac{\mu_k-1}{\mu_k}(x_2(t_k)-\hat{x}_{2,k})^T\hat{P}_{2,k}^{-1}(x_2(t_k)-\hat{x}_{2,k}).
		\end{align}
		Since we have
		$x_1(t_k) \in \mathcal{E}(\hat{x}_{1,k},\epsilon_{1,k}^2I_{n_1})$ and  $x_2(t_k) \in \mathcal{E}(\hat{x}_{2,k},\hat{P}_{2,k})$, (\ref{eq:24}) implies
		 that $x(t_k) \in \mathcal{E}(\hat{x}_k,\hat{P}_k)$. Additionally, applying Lemma \ref{lem:4} to to $\hat{P}_k$, one has $\mu_k = \sqrt{\frac{\mathrm{tr}(\hat{P}_{2,k})}{\mathrm{tr}(\epsilon_{1,k}^2I_{n_1})}}+1$, and the proof is completed.
	\end{proof}
	We summarize the proposed set-membership observer in Algorithm 1.
	\begin{algorithm}
		\begin{algorithmic}[1]
			\renewcommand{\algorithmicrequire}{\textbf{Input:}}
			\algnewcommand\algorithmicinput{\textbf{Design parameters:}}
			\algnewcommand\Param{\item[\algorithmicinput]}
			\renewcommand{\algorithmicensure}{\textbf{Output:}}
			\algnewcommand\algorithmicforeach{\textbf{for each}}
			\algdef{S}[FOR]{ForEach}[1]{\algorithmicforeach\ #1\ \algorithmicdo}
			\Require $\hat{x}_0,K_0,c_w(t),K_w(t)$
			\Ensure  $\mathcal{E}(\hat{x}_k,\hat{P}_k)$
			\If{$t=0$}
			\State Compute $P_1$ based on \cite{c15} and  \cite{c16}
			\State Use $P_1$ on $\hat{x}_0$ and $K_0$ to find $\hat{x}_{1,0}$, $\hat{x}_{2,0}$, $\epsilon_{1,0}$ and $\hat{P}_{2,0}$
			\Else
			\State Obtain $\hat{x}_1(t)$ and $\epsilon_1(t)$ based on (\ref{eq:10}) and (\ref{eq:12})
			
			\ForEach{$t = t_k$}
			\State Find $\alpha_k$ and $\gamma_k$ using (\ref{eq:91}) and (\ref{eq:92})
			\State Obtain $\hat{x}_{2,k|k-1}$ and $\hat{P}_{2,k|k-1}$ using (\ref{eq:50}) and (\ref{eq:51})
			\If{$G_k=0$}
			\State Set $\hat{x}_{2,k}$ and $\hat{P}_{2,k}$ to $\hat{x}_{2,k|k-1}$ and $\hat{P}_{2,k|k-1}$
			\Else
			\State Find $\beta_k$ using (\ref{eq:17})
			\State Obtain $\hat{x}_{2,k}$ and $\hat{P}_{2,k}$ using (\ref{eq:22}), (\ref{eq:23})
			\EndIf
			\State Obtain $\hat{x}_k$, $\hat{P}_k$ and $\mu_k$ using Theorem \ref{th:2}
			\EndFor
			\EndIf
		\end{algorithmic}
		\caption{Novel Set-Membership Observer for System $\Sigma$}
	\end{algorithm}

	\begin{remark}
		In a practical control system, such as in reachability application \cite{c59}, it is common that the system we consider is an uncertain nonlinear system, and we believe that the proposed approach can be extended to such a case. Adapting the approach from \cite{c19}, the idea is to linearize the nonlinear system about the state estimate and to bound the remainder term using interval mathematics. In this case, there are two steps of overapproximations: interval overapproximation of the ellipsoid set estimate (to apply the interval mathematics) and outer bounding ellipsoid of the interval of the remainder term (minimum volume or trace can be used as a metric). Then, we can implement our proposed approach to the linearized system with the bounded disturbance and remainder term.
	\end{remark}
	\section{Performance Analysis} \label{sec:pa}
	In this section, we would like to investigate some conditions under which the set estimate computed with Algorithm 1 will not become unbounded. The authors in \cite{c19, c20, c21} have already given some discussions on how to prove related statements, but they are not sufficient for our theoretical analysis, and we present the detailed analytical analysis in this section. Besides, it will be shown that our theoretical work will provide more relaxed conditions on the boundedness of $\hat{P}_k$.
	Hence, it is of interest to find positive definite matrices $\underline{P},\overline{P} \succ 0$ such that $\underline{P} \preceq \hat{P}_{k} \preceq \overline{P}$, for all $k \in \mathbb{N}$. First, let the following assumption be made.
	\begin{assumption}\label{ass:13}
		There exist positive constants $\underline{\alpha}$, $\overline{\alpha}$, $\underline{\beta}$, $\overline{\beta}$, $\underline{w}$, $ \overline{w}$ such that $\underline{\alpha}\leq \alpha_k \leq \overline{\alpha}$ and  $\underline{\beta} \leq \beta_k \leq \overline{\beta}$ for all $k \in \mathbb{N}$, and $\underline{w}I_{n_w} \leq K_w(t) \leq \overline{w}I_{n_w}$ for all $t \geq 0$.
	\end{assumption}
	Furthermore, let $b_2 = ||B_2'||$, $c_2 = ||C_2||$ and $d_2=\sigma_{min}(D_2')$. The first step in order to derive an upper bound for $\hat{P}_{2,k}$ is to bound the parameter $\gamma_k$, which determines the ellipsoid containing $u_2(t)$. The following lemma will state how to bound $\gamma_k$.
	\begin{lemma} \label{cor:2}
		Under Assumption \ref{ass:13}, the parameters $\gamma_k$ and $\frac{\gamma_k}{\gamma_k-1}$ are uniformly bounded, i.e., for some $\underline{\gamma}_1$, $\overline{\gamma}_1$, $\underline{\gamma}_2$ and $\overline{\gamma}_2$, $\underline{\gamma}_1 \leq \gamma_k \leq \overline{\gamma}_1$ and $\underline{\gamma}_2 \leq \frac{\gamma_k}{\gamma_k-1} \leq \overline{\gamma}_2$ are true for all $ k \in \mathbb{N}$.
	\end{lemma}
	\begin{proof}
		The proof is based on the definition of $\gamma_k$ in (\ref{eq:92}) and the fact that $\epsilon_{1,k}$ is uniformly bounded, and it is omitted for space.
	\end{proof}
	With the aforementioned results, we are ready to state the following lemma regarding the bounds of $\hat{P}_{2,k}$. 
	\begin{lemma} \label{lem:6}
		Under Assumptions \ref{ass:2} and \ref{ass:13}, $\hat{P}_{2,k}$ satisfies
		\begin{align}
			\underline{p}_{2}I_{n_2} \preceq \hat{P}_{2,k} \preceq \overline{p}_{2,k}I_{n_2}, \forall k\geq 0,
		\end{align}
		where
		\begin{gather} \label{eq:400}
			\begin{aligned}
					\underline{p}_2 &= \left(\frac{1-\underline{\beta}}{\underline{q}}+\frac{\overline{\beta}c_2^2}{d_2^2\min(\underline{\gamma}_1 \underline{\epsilon}_1^2,\underline{\gamma}_2\underline{w})}\right)^{-1},\\
				\overline{p}_{2,k} &= \left(\frac{\overline{f}}{1-\overline{\beta}}\right)^k\overline{p}_{2,0}+\frac{\overline{q}}{1-\overline{\beta}}\sum_{i=0}^{k-1}\left(\frac{\overline{f}}{1-\overline{\beta}}\right)^i, \\
				\underline{q} &= \frac{\kappa_1\kappa_2^2\Delta t \min(\baverage{\gamma}_1\baverage[2.5]{\epsilon}_1^2,\baverage{\gamma}_2\baverage[2.5]{w})}{1-\underline{\alpha}},\overline{p}_{2,0} = \norm{P_1K_0P_1^T}, \\
				\overline{f} &= \frac{\overline{a}_2^2e^{2\overline{\lambda}_2\Delta t}}{\underline{\alpha}},\overline{q} = \frac{\Delta t \max(\average{\gamma}_1\average[2.5]{\epsilon}_1^2,\average{\gamma}_2\average[2.5]{w})\overline{a}_2^2b_2^2}{2\overline{\lambda}_2(1-\overline{\alpha})}\left(e^{2\overline{\lambda}_2\Delta t}-1\right), 
			\end{aligned}
		\end{gather}
		and $\overline{\lambda}_2$ is chosen such that $\overline{\lambda}_2 > \max{Re(\lambda(A_4))}$, and for some $\overline{a}_2,\kappa_1,\kappa_2>0$. 
	\end{lemma}
	\begin{proof}
		See Appendix \ref{app:8}.
	\end{proof}
	Note that Lemma \ref{lem:6} only proves that there exists a time-dependent upper bound for $\hat{P}_{2,k}$, but it is more important to find a uniform upper bound of $\hat{P}_{2,k}$. Obviously, the parameter $\overline{f}$, which characterizes the exponential stability of the weakly unobservable subsystem $\Sigma_2$, is a key factor in determining the uniform upper bound of $\hat{P}_{2,k}$. There are two cases corresponding to different uniform bounds of $\hat{P}_{2,k}$.
	Now, the first case that leads to a uniform bound on $\hat{P}_{2,k}$ is described in the following lemma.
	\begin{lemma} \label{lem:13}
		Under Assumptions \ref{ass:2} and \ref{ass:13}, if $\overline{f} <1-\overline{\beta}$, then $\hat{P}_{2,k}$ is uniformly bounded above, i.e. $\hat{P}_{2,k} \preceq \overline{p}_2I_{n_2}, \forall k \in \mathbb{N}$, where 
		$
			\overline{p}_2 = \frac{\overline{f}\overline{p}_{2,0}}{1-\overline{\beta}}+\frac{\overline{q}}{1-\overline{\beta}-\overline{f}} .
		$
	\end{lemma}
	\begin{proof}
		Based on Lemma \ref{lem:6}, one has
		\begin{equation}
				\overline{p}_{2,k} \leq \left(\frac{\overline{f}}{1-\overline{\beta}}\right)^k\overline{p}_{2,0}+\frac{\overline{q}}{1-\overline{\beta}}\sum_{i=0}^{\infty}\left(\frac{\overline{f}}{1-\overline{\beta}}\right)^i \leq \overline{p}_2
		\end{equation} 
		and the proof is completed.
	\end{proof} 
	\begin{remark} \label{rem:3}
		After some manipulations, we see that Lemma \ref{lem:13} can be implemented if 
		\begin{equation} \label{eq:90}
			\overline{\lambda}_2 < \frac{ln(1-\overline{\beta}) + ln(\underline{\alpha})-2ln(\overline{a}_2)}{2\Delta t}
		\end{equation}
		is satisfied. Since in general the parameter $\overline{\beta}$ and $\underline{\alpha}$ are approximately equal to 1 and 0, respectively, we conclude that $A_4$ needs to be sufficiently stable, i.e., $\max Re(\lambda(A_4)) < - \lambda^*$ for some $\lambda^* > 0$, in order to use Lemma \ref{lem:13}.  
	\end{remark}
			\begin{figure*}[b!]
		\centering
		\hrule
		\subfloat[$x_1$]{\includegraphics[clip,trim = 2.8cm 11.2cm 6.8cm 6.5cm,width=\columnwidth-5.25cm]{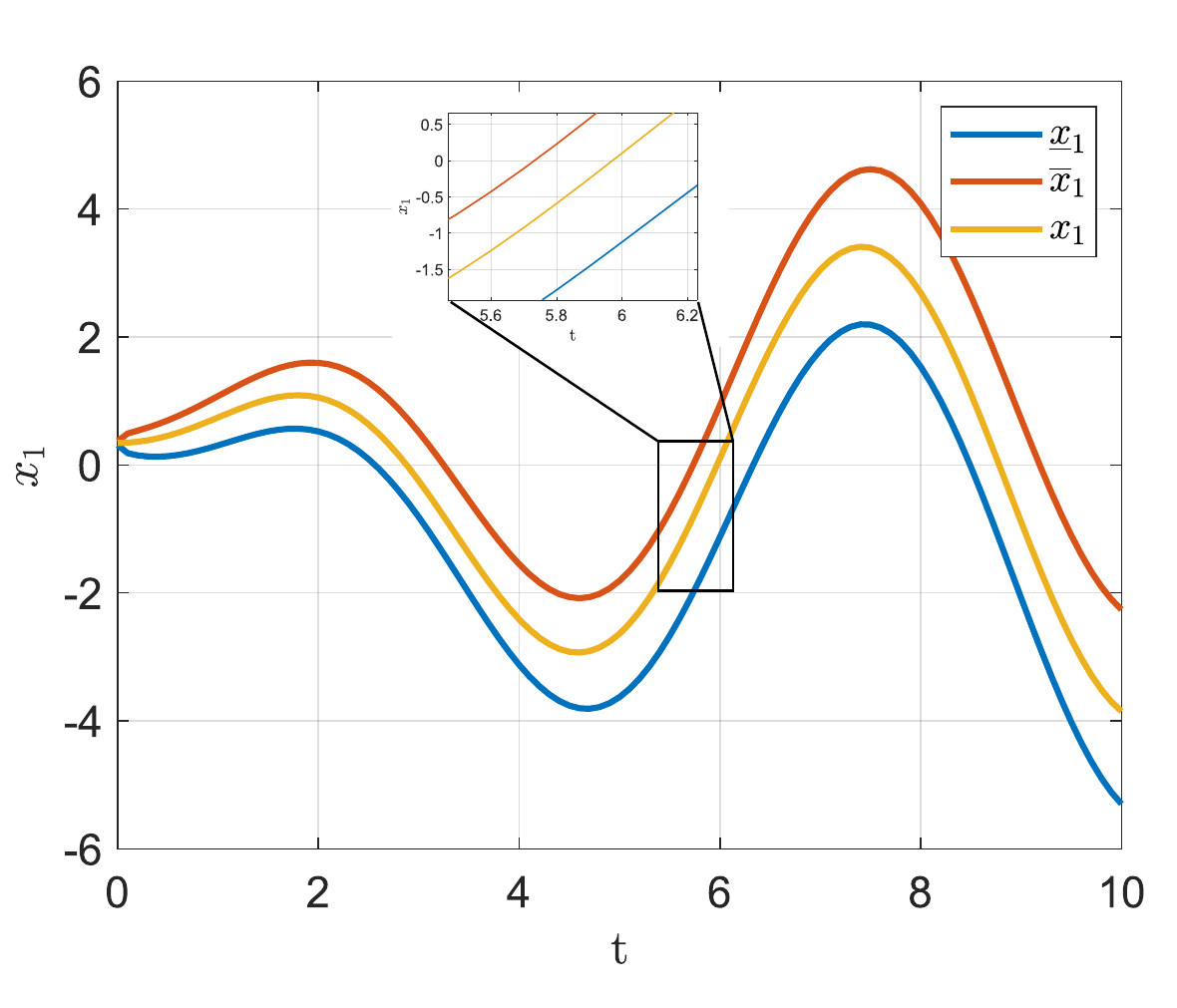}}
		\subfloat[$x_2$]{\includegraphics[clip,trim = 5cm 12.1cm 4.4cm 6.3cm,width=\columnwidth-5.25cm]{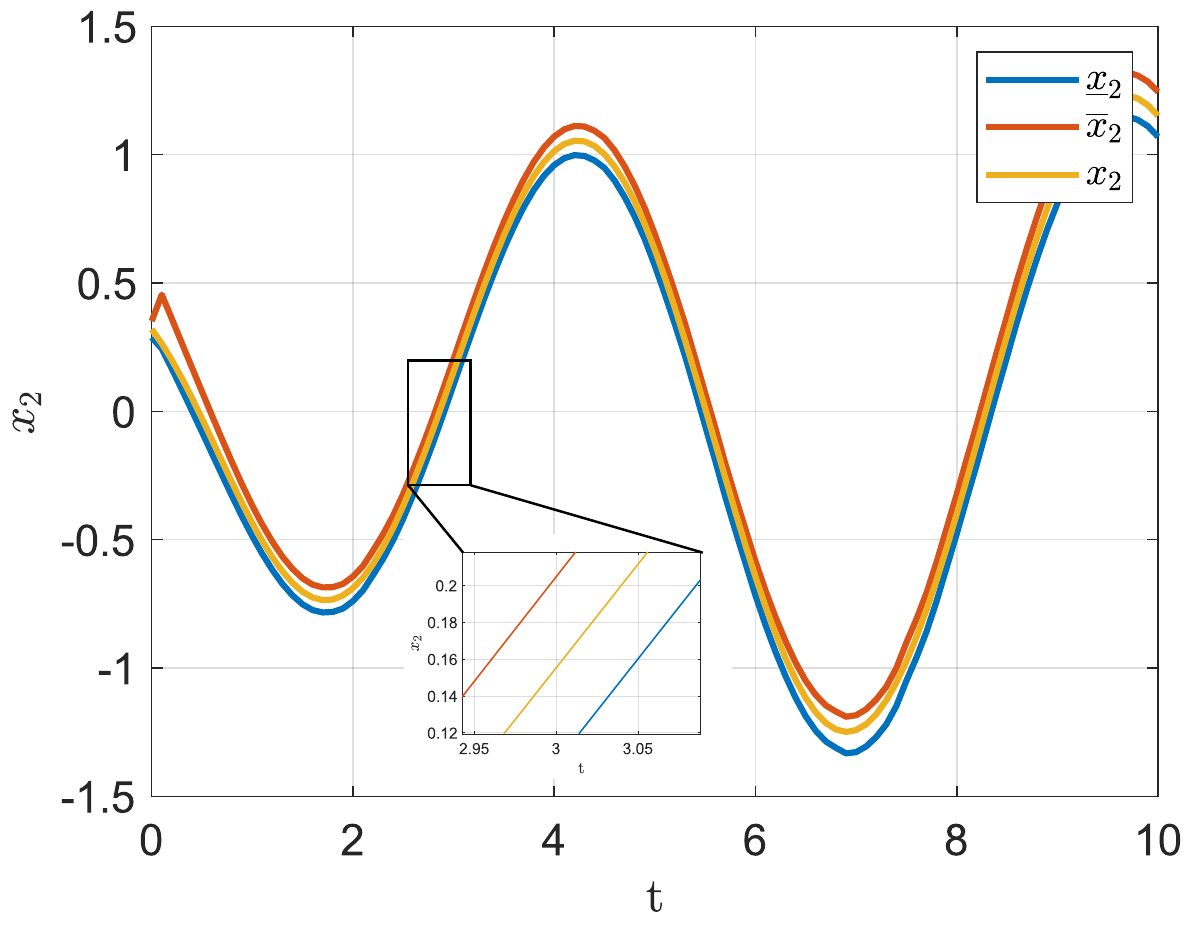}}
		\subfloat[$x_3$]{\includegraphics[clip,trim = 3.9cm 12.1cm 5.9cm 6.5cm,width=\columnwidth-5.25cm]{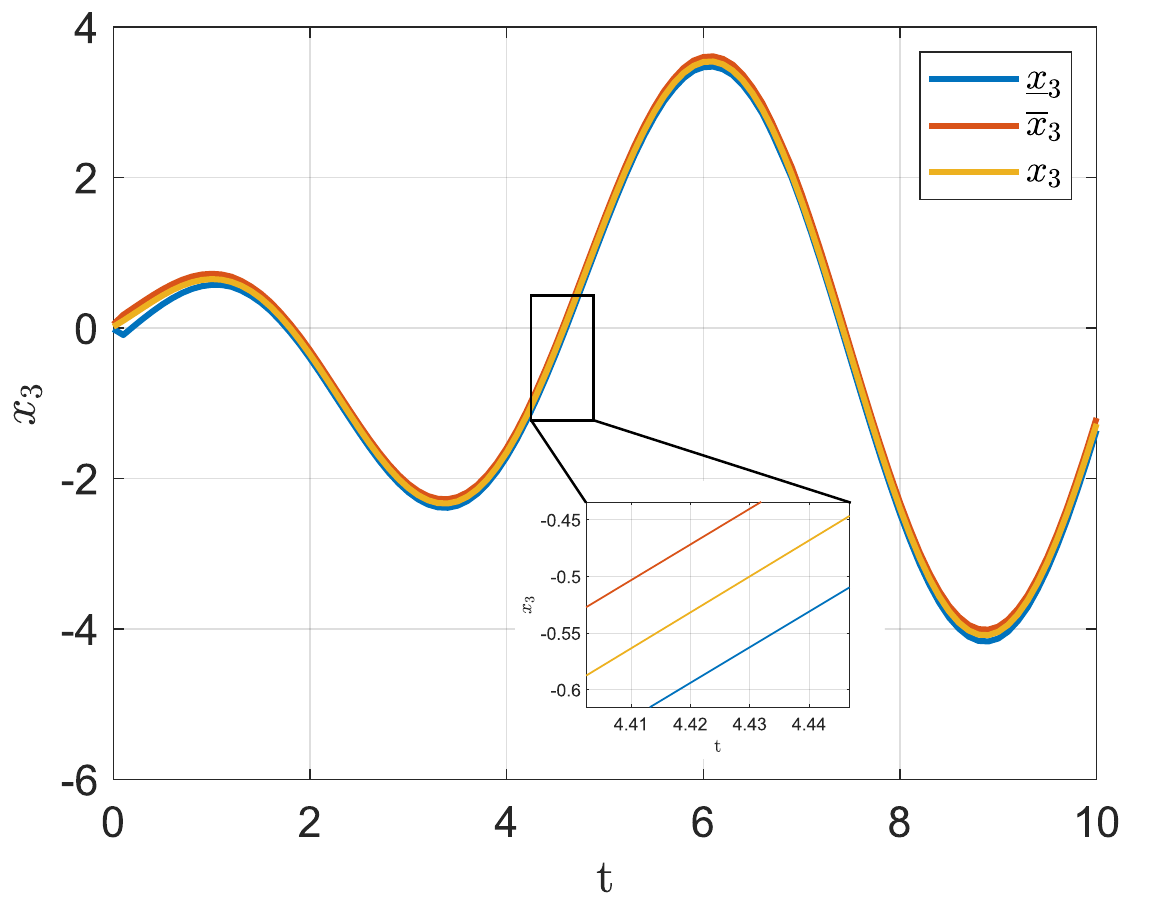}}
		\subfloat[$x_4$]{\includegraphics[clip,trim = 3.8cm 12.4cm 5.8cm 6.2cm,width=\columnwidth-5.25cm]{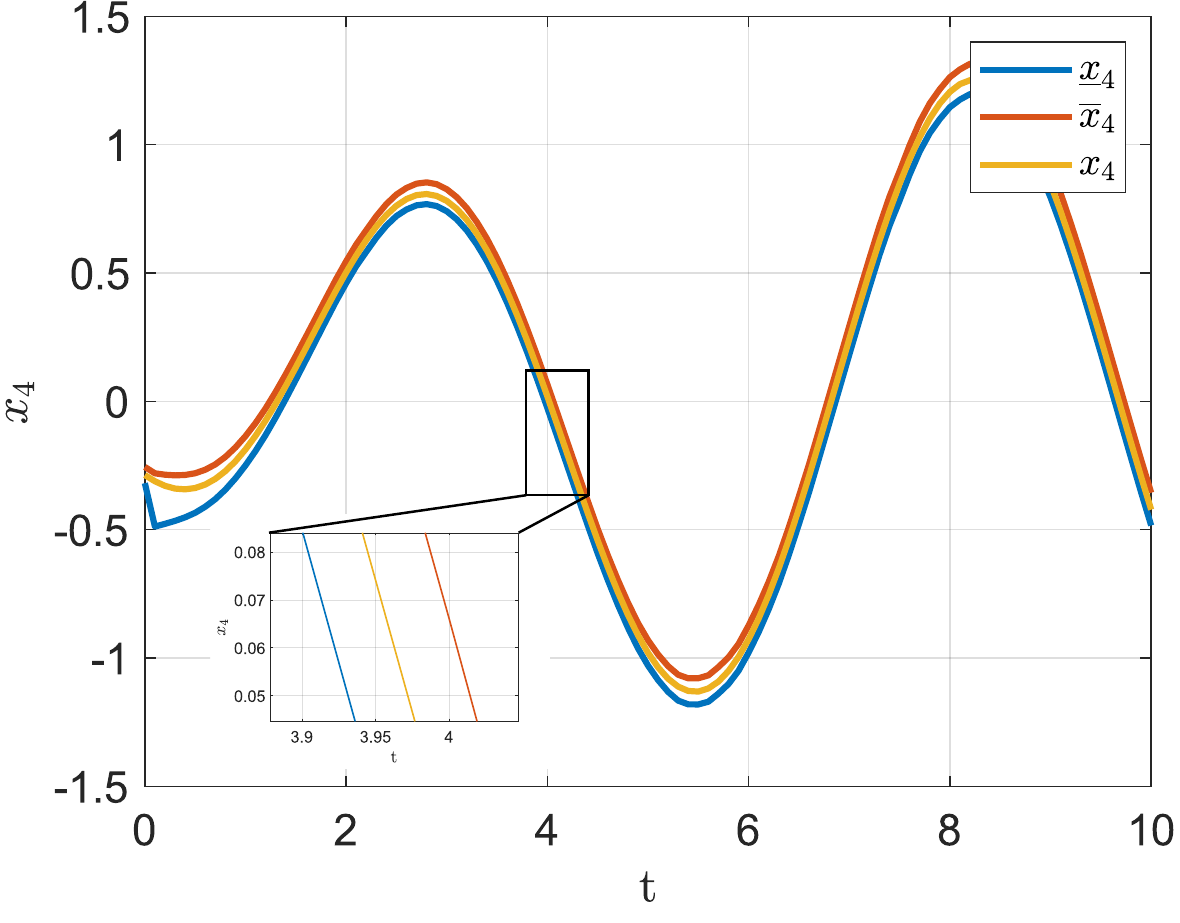}}
		\subfloat[$x_5$]{\includegraphics[clip,trim = 3.8cm 11.8cm 5.5cm 6.5cm,width=\columnwidth-5.25cm]{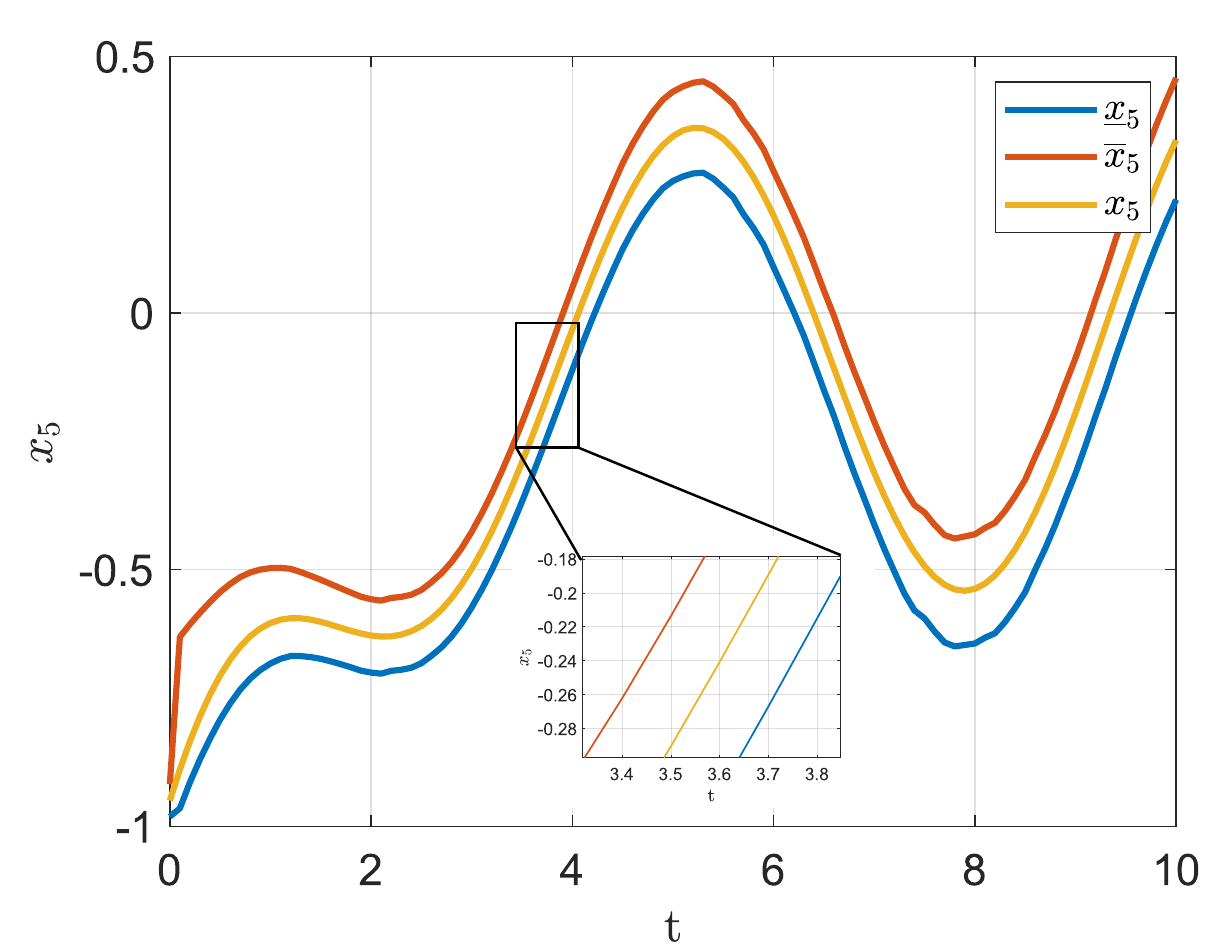}}
		\caption{Volume and bounds of ellipsoidal estimate for Example 1}
		\label{fig:ex1}
	\end{figure*}
	
	In order to prove the other case, the following assumption needs to be made about subsystem $\Sigma_2$. 
	\begin{assumption}\label{ass:14}
		There exist an integer $r \in \mathbb{N}$ and a constant $\underline{\rho}$ with $0 < \underline{\rho} < \infty$ such that
		\begin{equation}
			\underline{\rho}I_{n_2} \preceq \sum_{i=k-r}^k
			e^{-(k-i)A_4^T\Delta t}C_2^TG_{i}^{-1}C_2e^{-(k-i)A_4\Delta t}
		\end{equation}
		for all $k \geq r$.
	\end{assumption}

	\begin{remark}
		This assumption is similar to uniform observability \cite{c57}, but it is more relaxed. Indeed, our analysis shows that we only need to know the lower bound on the observability Grammian, which is not the case in \cite{c20} and \cite{c21}, where assumptions on uniform controllability and observability as well as full knowledge on the bounds of the observability Grammian are required. Thus, more flexibilities are brought with our proposed algorithm. Additionally, if $\Sigma_2$ is observable, Assumption \ref{ass:14} is trivially satisfied. 
	\end{remark}

	Finally, the following lemma illustrates how to find a uniform upper bound of $\hat{P}_{2,k}$ for the other case.
	\begin{lemma} \label{lem:12}
		Under Assumptions \ref{ass:2}, \ref{ass:13} and \ref{ass:14}, if $\overline{f} \geq 1-\overline{\beta}$, then $\hat{P}_{2,k}$ is uniformly bounded above, i.e., for all $k \in \mathbb{N}$, $ \hat{P}_{2,k} \preceq \overline{p}_2I_{n_2}$, where 
		\begin{equation} \label{eq:405}
			\begin{aligned}
				\overline{p}_2 &= \max_{1 \leq k \leq r}\left(\overline{p}_{2,k}, \frac{1}{\underline{\beta}((1-\overline{\beta})\varphi)^r\underline{\rho}}\right),\\
				\varphi &= \left(1+\frac{\underline{a}_2^2\average{q}e^{2\underline{\lambda}_2\Delta t}\average[2.5]{\alpha}}{\baverage{p}_2}\right)^{-1}\underline{\alpha},
			\end{aligned}			
		\end{equation} where $\underline{\lambda}_2$ is chosen such that $\underline{\lambda}_2 > \max Re(\lambda(-A_4))$, and for some $\underline{a}_2>0$.
	\end{lemma}
	\begin{proof}
		See Appendix \ref{app:12}.
	\end{proof}
	\begin{remark}
		Depending on the value of $\overline{f}$, one can use either Lemma \ref{lem:13} or Lemma \ref{lem:12}. If $\overline{f}$ satisfies the condition given in Lemma \ref{lem:13}, we do not require Assumption \ref{ass:14} to be satisfied, while the methods in \cite{c19} and \cite{c20} require the uniform observability assumption at all given conditions to ensure that $\hat{P}_{2,k}$ is uniformly bounded. Hence, our proposed algorithm provides an alternative in case the condition of observability is not satisfied for subsystem $\Sigma_2$. 
%
	\end{remark}
	
	The previous corollaries and lemmas clearly show that both ellipsoids $\mathcal{E}(\hat{x}_{1,k},\epsilon_{1,k}^2 I_{n_1})$ and $\mathcal{E}(\hat{x}_{2,k},\hat{P}_{2,k})$ are uniformly bounded. Therefore, the uniform bound on $\hat{P}_k$ can be concluded through the following theorem.
	\begin{theorem} \label{thm:2}
		Under Assumptions \ref{ass:2}, \ref{ass:10}, \ref{ass:13} and \ref{ass:14}, $\hat{P}_k$ is uniformly bounded above and below, i.e., there are positive definite matrices $\underline{P}, \overline{P}$ such that $\underline{P} \preceq \hat{P}_k \preceq \overline{P}, \forall k \in \mathbb{N}$.
	\end{theorem}

	\begin{proof}
		First, we have to find the bounds on the parameter $\mu_k$. Based on the fact that $\epsilon_{1,k}$ is uniformly bounded and Lemma \ref{lem:13} (if $\overline{f} < 1-\overline{\beta}$) or Lemma \ref{lem:12} (if $\overline{f} \geq 1-\overline{\beta}$), there exist positive constants $\underline{\mu}_1$, $\overline{\mu}_1$, $\underline{\mu}_2$ and $\overline{\mu}_2$ such that $\underline{\mu}_1 \leq \mu_k \leq \overline{\mu}_1$ and $\underline{\mu}_2 \leq \frac{\mu_k}{\mu_k-1} \leq \overline{\mu}_2$, $\forall k \in \mathbb{N}$.
		Therefore, it is clear that $\hat{P}_k$ is uniformly bounded below and above by 
		$
			\underline{P} = P_1^{-1}\mathrm{diag}(
				\underline{\mu}_1\underline{\epsilon}_1^2I_{n_1},\underline{\mu}_2\underline{p}_2I_{n_2}
			)P_1^{-T}
		$ and
		$
			\overline{P} = P_1^{-1}\mathrm{diag}(
				\overline{\mu}_1\overline{\epsilon}_1^2I_{n_1},\overline{\mu}_2\overline{p}_2I_{n_2}
			)P_1^{-T},
		$
	respectively, and the proof is completed.
	\end{proof}
	\begin{remark} \label{rem:6}
		Note that there are some scenarios in which we can find a uniform bound for $\hat{P}_k$ using our proposed algorithm, while we cannot find such bound using the algorithms in \cite{c20} and \cite{c21}.  
		First, let us assume the system $\Sigma$ is unstable. Then, if we apply the algorithms in \cite{c20} and \cite{c21}, the boundedness of $\hat{P}_k$ is not guaranteed. Nevertheless, if the system $\Sigma$ is still unstable, but the weakly unobservable subsystem $\Sigma_2$ is sufficiently stable based on Remark \ref{rem:3}, then our proposed algorithm provides a bounded set estimate, and an illustrative example of this scenario will be given in Section \ref{sec:uns}.
	\end{remark}

	\section{Numerical Simulations}
		In this section, we demonstrate our algorithm's performance via two illustrative numerical examples. The first example shows the application of out set-membership observer to an aircraft, while the second one is mainly to illustrate our claim in Remark \ref{rem:6}.
	\subsection{Example 1: Aircraft Dynamics}
	First, a 5-th order linear system representing the lateral axis model of an L-1011 fixed-wing aircraft is borrowed from \cite{c47}. The system dynamics is given as $\Sigma = (A,B,C,D)$ with
	\begin{gather*}
		\begin{aligned}
			\allowdisplaybreaks
			A &= \begin{bsmallmatrix}
				0&0&1&0&0\\0&-0.154&-0.0042&1.54&0\\0&0.2490&-1&-5.2&0\\0.0386&-0.996&-0.003&-0.117&0\\0&0.5000&0&0&-0.5
			\end{bsmallmatrix},
			&C &= \begin{bsmallmatrix}
				0 &1& 0 &0 &-1\\0&0&1&0&0\\0&0&0&1&0\\0&0&0&0&0
			\end{bsmallmatrix},\\
			B &= \begin{bsmallmatrix}
				0&0\\-0.7440&-0.0320\\0.3370&-1.1200\\0.0200&0\\0&0
			\end{bsmallmatrix},
			&D &= \begin{bsmallmatrix}
				1&1\\1&1\\1&1\\1&1
			\end{bsmallmatrix}.
			\allowdisplaybreaks
		\end{aligned}
	\end{gather*}
Additionally, the inputs in our example are assumed to be UBB, namely $w(t)$ is bounded by $\mathcal{E}(c_w(t),K_w(t))$ with $c_w(t) = col(
	0.5sin(t),0.4cos(t))$ and $K_w(t) = diag(3,5)$, and the initial condition is bounded by $\mathcal{E}(\hat{x}_0,K_0)$ with $\hat{x}_0 = col(
	0.342,0.32,0.0178,-0.287,-0.9497
	)$ and $K_0 = 0.001I_5$.  
		
	In the numerical experiment, the unknown input $w(t)$ is simulated as $col(0.8sin(t),0.7cos(t))$, and the step size $\Delta t$ is chosen to be 0.1. Here, we are ready to use Algorithm 1 to obtain a robust state estimation set for system $\Sigma$. After implementing Algorithm 1,
	\begin{figure}[h!]
		\centering
				\includegraphics[clip,trim = 3.5cm 8cm 3.5cm 8.5cm,width=\columnwidth-2.5cm]{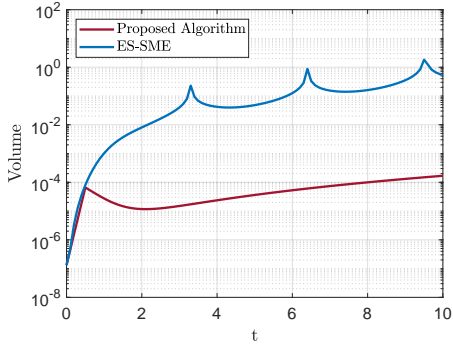}
		\caption{Volume of bounding ellipsoids for Example 1}
		\label{fig:vol}
	\end{figure}
the result we get is compared with the ES-SME algorithm developed by Liu et al. in \cite{c21}, where the authors showed that their approach performs better than other set-membership algorithms, such as AESMF \cite{c48}, BA-SME \cite{c20}, and Kalman filter. 
	
	Figure \ref{fig:vol} shows the evolution of the volume of the bounding ellipsoids computed with both algorithms over time.
	It is clear that our proposed algorithm outperforms the ES-SME algorithm in which the estimate resulted from the ES-SME algorithm diverges.  Furthermore, we are able to show that our set estimate bounds the true state at all time as illustrated in Figure \ref{fig:ex1}(a)-(e). In Figure \ref{fig:ex1}(a)-(e), $\overline{x}_i$ and $\underline{x}_i$ denote the estimated upper and lower bound of the system states provided by Algorithm 1, respectively.

		\subsection{Example 2: Unstable System} \label{sec:uns}
	In this example, we would like to validate the claim we made in Remark \ref{rem:6} by creating an illustrative system such that the system $\Sigma$ is unstable but the weakly unobservable subystem is stable. The system dynamics is given as $\Sigma=(A,B,C,D)$ with
	 
		\begin{gather*}
		A = \begin{bsmallmatrix}
			2 &1 &1\\0& -17& 0\\0 &0 &-20
		\end{bsmallmatrix},
		B = \begin{bsmallmatrix}
			1&1\\0&1\\1&1
		\end{bsmallmatrix},
		C = \begin{bsmallmatrix}
			1&0&0
		\end{bsmallmatrix},
		D = 0.
	\end{gather*}
	The ellipsoid $\mathcal{E}(c_w(t),K_w(t))$ is the same as in Example 1, and the ellipsoid $\mathcal{E}(\hat{x}_0,K_0)$ is described by $\hat{x}_0 = col(0.03,0.03,0.03)$ and $K_0 = diag(0.01,0.01,0.01)$. For this example, the system $\Sigma$ does not need to be transformed,
 i.e., $P_1 = I_3$.
		 As mentioned previously, the strongly observable subsystem is unstable with the eigenvalue of $2$, while the weakly unobservable subsystem is stable with the eigenvalues of $-17$ and $-20$. 
		For this example, the parameters $\overline{\alpha}$, $\underline{\alpha}$, $\overline{\beta}$ and $\underline{\beta}$ are found to be 0.9 , 0.1, 0 and 0, respectively ($\beta_k$ is always zero since measurement step is not used in this case). Therefore, by Lemma \ref{lem:13}, we can conclude that the set estimate is uniformly bounded based on Theorem \ref{thm:2}.
	\renewcommand{\thefigure}{5}
	\begin{figure}[h!]
		\centering
		\subfloat[Volume]{	\includegraphics[clip,trim = 3.5cm 8cm 4cm 9cm,width=\columnwidth/2-0.2cm]{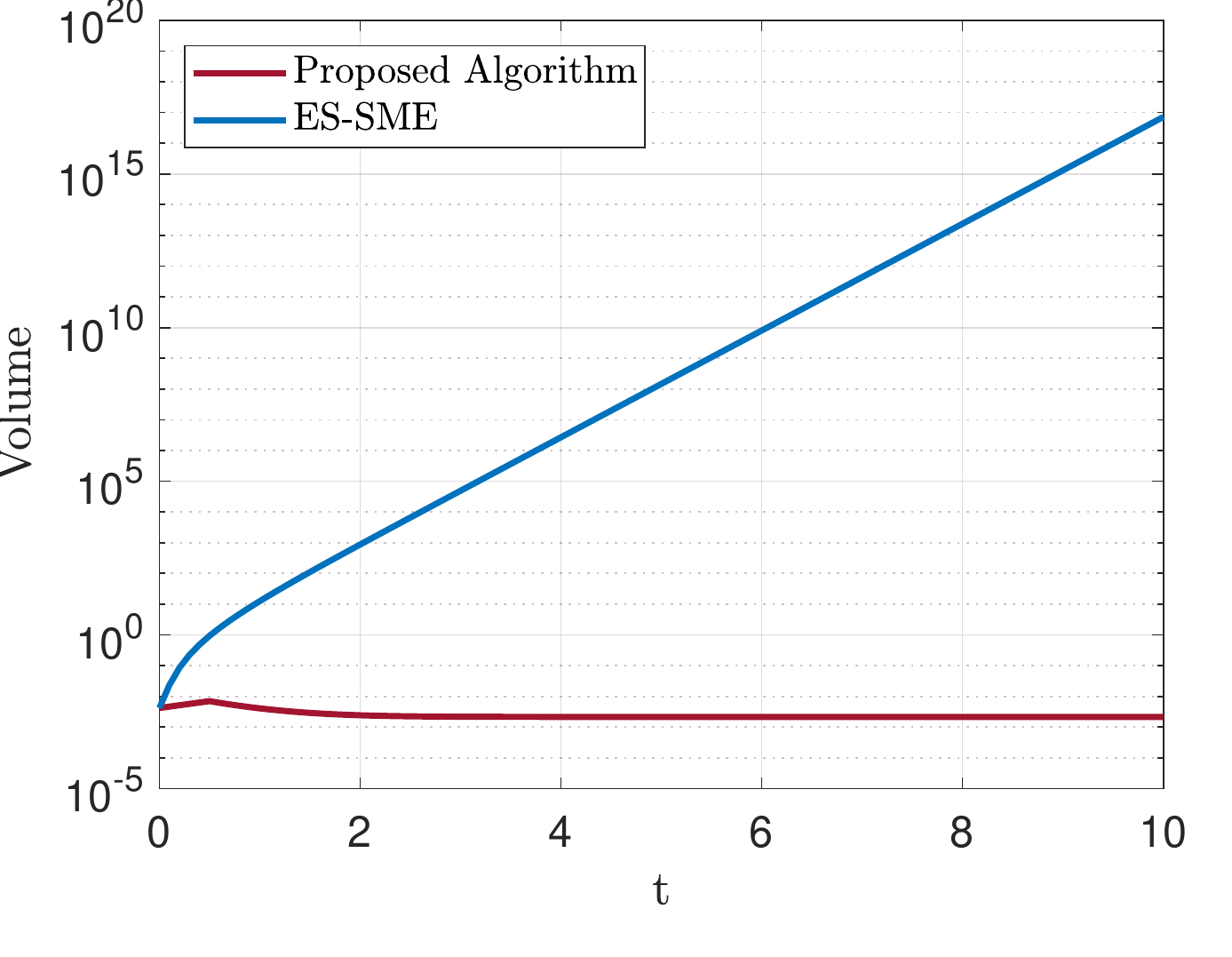} 
			\label{fig:vol2}}
		\subfloat[Projected bounding ellipsoids on $x_2$ and $x_3$]{\includegraphics[clip,trim =3.5cm 8cm 3.8cm 8.9cm,width=\columnwidth/2-0.2cm]{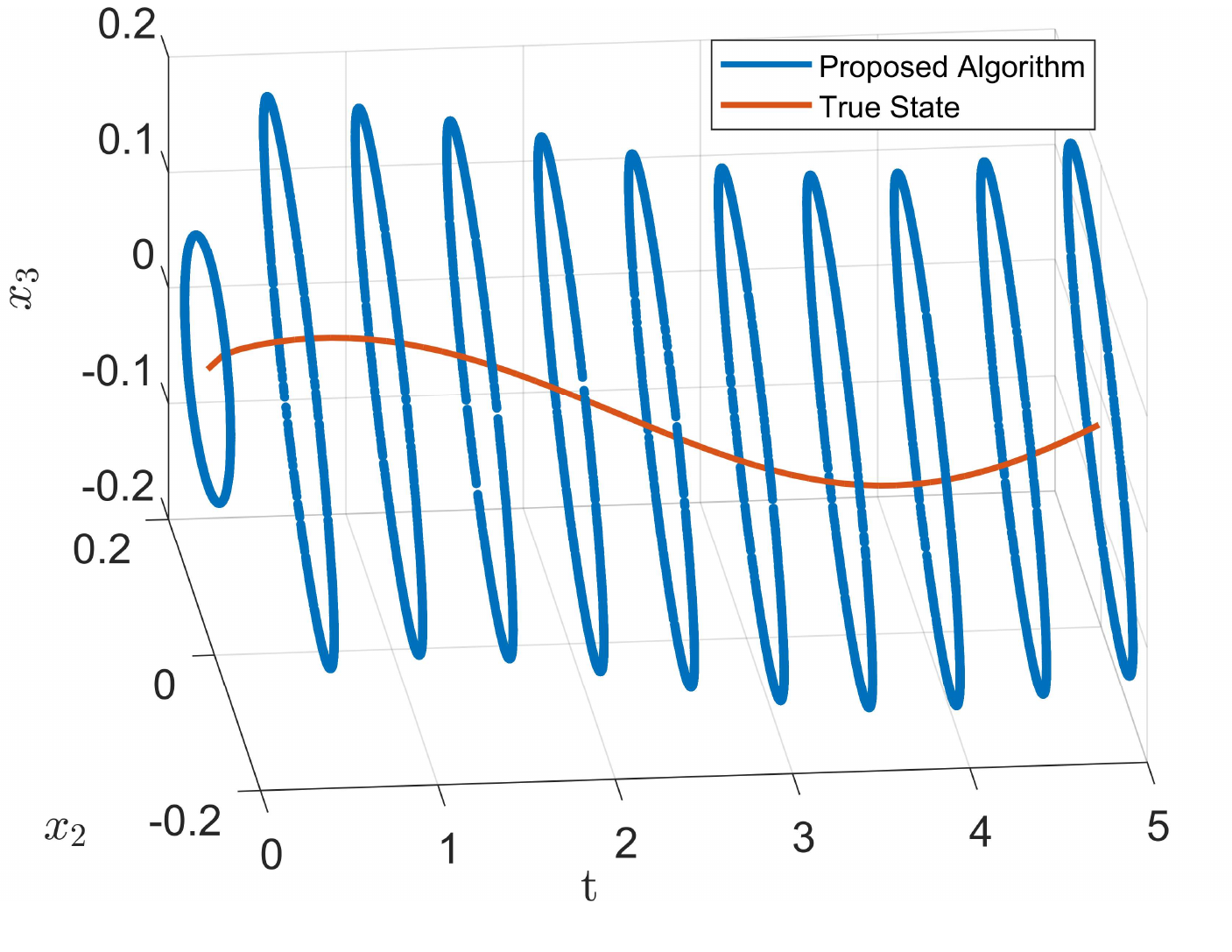}
			\label{fig:x2x3}}
		\caption{Volume and bounding ellipsoids for Example 2}
		\label{fig:ex2}
	\end{figure}
	
	Figure \ref{fig:ex2}(a) presents the volumes of the ellipsoids computed with both proposed and ES-SME algorithm, while Figure \ref{fig:ex2}(b) shows the projected ellipsoids in the $x_2$-$x_3$ plane containing the true state trajectory for this example.
	Clearly, the ES-SME algorithm's set estimate becomes unbounded, whereas our proposed algorithm's set estimate still provides a sequence of uniformly bounded ellipsoids, which validates our claim that the proposed algorithm can still provide an accurate set estimate even though $\Sigma$ is unstable.

	\section{Conclusion}
	In this paper, we have presented a novel technique to set-membership state estimation for a dynamical system with unknown-but-bounded exogenous inputs. We have analytically proved the boundedness of the set estimates and discussed some important properties. Simulation results have shown that the proposed algorithm has better performance compared with some of the existing set-membership algorithms in terms of the accuracy and numerical stability.

	\appendices
	
	\section{Proof of Proposition \ref{prop:1}} \label{app:5}
	Define the estimation error of the strongly observable subsystem $\Sigma_1$ as $e_1(t) \triangleq x_1(t)-\hat{x}_1(t)$. With (\ref{eq:8}) as well as the fact that $E=A_1-F\mathcal{O}_l$ and $F\mathcal{G}_l = \begin{bmatrix}B_1'&0\hdots&0\end{bmatrix}$, the error dynamics can be expressed as
	\begin{equation}\label{eq:9}
	\begin{aligned}
		\dot{e}_1&=Ee_1+F(z_{0:l}-\hat{z}_{0:l}).
		\end{aligned}
	\end{equation}
	Integrating and taking the norm of (\ref{eq:9}), one has
	\begin{equation} \label{eq:95}
	\begin{aligned}
		\norm{e_1(t)} &\leq \norm{e^{Et}}.\norm{e_1(0)}\\ 
		+ &\int_0^t \norm{e^{E(t-\tau)}}.\norm{F}.\norm{z_{0:l}(\tau)-\hat{z}_{0:l}(\tau)}d\tau
	\end{aligned}
	\end{equation}
	where $e_1(0) \triangleq x_1(0)-\hat{x}_1(0)$. First, we would like to bound $|\tilde{z}_{i,k}|$. Adapting the proof of Lemma 1 in \cite{c36}, we will arrive at 
	\begin{equation} \label{eq:300}
		|\tilde{z}_{i,k}| \leq \epsilon^{l-k}\delta+\left(\frac{K\sqrt{l+1}}{\epsilon^{k}}\norm{\tilde{z}_i(0)}-\epsilon^{l-k}\delta\right)e^{-\frac{at}{\epsilon}}.
	\end{equation}
	Since $\epsilon \in (0,1)$, using the knowledge on $\overline{z}_0$, (\ref{eq:300}) and Assumption \ref{ass:2}, (\ref{eq:95}) can be expressed as
	\begin{equation} \label{eq:200}
		\norm{e_1(t)} \leq \norm{e^{Et}} \norm{P_1K_0P_1^T}^{\frac{1}{2}} + \norm{F}\sqrt{n_y(l+1)} \Psi(t),
	\end{equation}
	where $\Psi(t)$ is defined in (\ref{eq:DefPsi}).
	Therefore, (\ref{eq:200}) can be expressed as $\norm{e_1(t)} \leq \epsilon_{1}(t)$
	where $\epsilon_{1}(t)$ is given in (\ref{eq:12}), and the proof is completed. 
	
	\section{Proof of Lemma \ref{lem:4}} \label{app:2}
	First, suppose that $x = \mathrm{col}(x_{q_1},x_{q_2})$. Then, let us write $(x-\hat{x}_q)^TQ^{-1}(x-\hat{x}_q)$ as
	$
		\frac{1}{g}(x_{q_1}-\hat{x}_{q_1})^TQ_1^{-1}(x_{q_1}-\hat{x}_{q_1})
		+\frac{g-1}{g}(x_{q_2}-\hat{x}_{q_2})^TQ_2^{-1}(x_{q_2}-\hat{x}_{q_2}).
	$
	Since $x_{q_1} \in \mathcal{E}(\hat{x}_{q_1},Q_1)$ and $x_{q_2} \in \mathcal{E}(\hat{x}_{q_2},Q_2)$, we will see that $x \in \mathcal{E}(\hat{x}_q,Q)$, and the first part of the proof is done. 
	To find the parameter $g$, first $\mathrm{tr}(Q)$ can be expressed as
	$
		\mathrm{tr}(Q) =g\cdot \mathrm{tr}(Q_1)+\frac{g}{g-1}\mathrm{tr}(Q_2) 
	$
	Minimizing $\mathrm{tr}(Q)$ with respect to $g$, one has
	$
		g = \sqrt{\frac{\mathrm{tr}(Q_2)}{\mathrm{tr}(Q_1)}}+1,
	$
	and the second part of the proof is completed.

	\section{Proof of Lemma \ref{lem:6}} \label{app:8}
	First, we have
	$
		K_u(t) \preceq \max(\average{\gamma}_1\average[2.5]{\epsilon}_1^2,\average{\gamma}_2\average[2.5]{w})I_{n_1+n_w} 
	$  by Lemma \ref{cor:2}.
	We claim that $\hat{P}_{2,k} \preceq \overline{p}_{2,k}I_{n_2}$, where $\overline{p}_{2,k} = \overline{f}\overline{p}_{2,k-1}+\overline{q}$. To see it, we proceed by induction. First, note that for any $\lambda_2 > \max_{\lambda(A_4)} Re(\lambda(A_4))$, there exists $\overline{a}_2 > 0$ such that $||e^{A_4\Delta t}|| \leq \overline{a}_2e^{\overline{\lambda}_2 \Delta t}, \forall \Delta t \geq 0$. When $k=1$, it is true that $\hat{P}_{2,1} \preceq \overline{f}\overline{p}_{2,0} + \overline{q}$, where $\overline{p}_{2,0} = \norm{P_1K_0P_1^T}$. Suppose that $\hat{P}_{2,k} \preceq \overline{p}_{2,k}I_{n_2}$.
	Using the upper bound on $K_u(t)$ and (\ref{eq:51}), we have
	\begin{align} \label{eq:56}
		\hat{P}_{2,k+1|k}
		&\preceq\frac{\overline{a}_2^2e^{2\overline{\lambda}_2\Delta t}\overline{p}_{2,k}}{\underline{\alpha}}I_{n_2}\nonumber\\ &\quad+ \frac{\Delta t  \max(\average{\gamma}_1\average[2.5]{\epsilon}_1^2,\average{\gamma}_2\average[2.5]{w})\overline{a}_2^2b_2^2}{2\overline{\lambda}_2(1-\overline{\alpha})}(e^{2\overline{\lambda}_2\Delta t}-1)I_{n_2}.
	\end{align}
	Based on the Woodbury matrix identity, the next step is to rewrite (\ref{eq:23}) into
	\begin{equation} \label{eq:55}
		\hat{P}_{2,k+1} = \left((1-\beta_{k+1})\hat{P}_{2,k+1|k}^{-1}+\beta_{k+1}C_2^TG_{k+1}^{-1}C_2\right)^{-1}.
	\end{equation}
	Then, using (\ref{eq:56}), we have
	\begin{align*}
		\hat{P}_{2,k+1} \preceq \frac{1}{1-\beta_{k+1}}\hat{P}_{2,k+1|k} \preceq \frac{\overline{f}\overline{p}_{2,k}+\overline{q}}{1-\overline{\beta}}I_{n_2} = \overline{p}_{2,k+1}I_{n_2}.
	\end{align*}
	As $\overline{p}_{2,k+1}$ is iteratively defined by $\frac{\overline{f}\overline{p}_{2,k}+\overline{q}}{1-\overline{\beta}}$, then $\overline{p}_{2,k} = \frac{\overline{f}\overline{p}_{2,k-1}+\overline{q}}{1-\overline{\beta}}$.  To find the lower bound $\underline{p}_2$,
	 by Lemma \ref{cor:2} and letting $\int_0^{\Delta t}e^{A_4s}e^{A_4^Ts}ds \succeq \kappa_1 I_{n_2} $ for some $\kappa_1 >0$, (\ref{eq:51}) can be written as
	\begin{equation} \label{eq:58}
			\resizebox{.95\hsize}{!}{$\hat{P}_{2,k|k-1}
			\succeq \frac{e^{A_4\Delta t}\hat{P}_{2,k-1}e^{A_4^T\Delta t}}{\alpha_k}+ \frac{\kappa_1\kappa_2^2\Delta t \min(\baverage{\gamma}_1\baverage[2.5]{\epsilon}_1^2,\baverage{\gamma}_2\baverage[2.5]{w})}{1-\underline{\alpha}}I_{n_2},$}
	\end{equation}
	where $\kappa_2 = \sigma_{min}(B_2')$.
	Next, we have
	$
		\hat{P}_{2,k|k-1}^{-1} \preceq \frac{1}{\underline{q}}I_{n_2},
	$
	and it can be shown that $0 \preceq C_2^TG_k^{-1}C_2 \preceq \frac{c_2^2}{d_2^2\min(\underline{\gamma}_1 \underline{\epsilon}_1^2,\underline{\gamma}_2\underline{w})}I_{n2}$, and (\ref{eq:55}) can be expressed as
	$	\hat{P}_{2,k} \succeq \underline{p}_2I_{n_2}$, where $\underline{p}_2$ is defined in (\ref{eq:400}),
	and the proof is completed.

	\section{Proof of Lemma \ref{lem:12}} \label{app:12}
	
	\textit{Case 1:}
	If $0 < k \leq r$, according to Lemma \ref{lem:6}, one can easily choose an upper bound $\overline{p}_2' \triangleq \max_{1 \leq k \leq r}\overline{p}_{2,k}$ such that $\hat{P}_{2,k} \preceq \overline{p}_2', \forall 0<k \leq r$.\\
	\textit{Case 2:}
	When $k>r$, adapting Lemma 2 in \cite{c57}, (\ref{eq:55}) can be expressed as
	\begin{equation} \label{eq:87}
		\hat{P}_{2,k}^{-1} \succeq (1-\overline{\beta})\varphi e^{-A_4^T\Delta t}\hat{P}_{2,k-1}^{-1}e^{-A_4\Delta t}+\underline{\beta}C_2^TG_k^{-1}C_2,
	\end{equation} where $\varphi$ is defined in (\ref{eq:405}).
	Doing recursive iteration to (\ref{eq:87}), one has
	{\small
	\begin{align}\label{eq:88}
		&\hat{P}_{2,k}^{-1} \succeq (1-\overline{\beta})^{r+1}\varphi^{r+1}e^{-(r+1)A_4^T\Delta t}\hat{P}_{2,k-r-1}^{-1}e^{-(r+1)A_4\Delta t} \nonumber\\&\quad+\sum_{i=k-r}^k \underline{\beta}((1-\overline{\beta})\varphi)^{k-i}e^{-(k-i)A_4^T\Delta t}C_2^TG_i^{-1}C_2 e^{-(k-i)A_4\Delta t}.
	\end{align}}
	Based on Assumption \ref{ass:14}, (\ref{eq:88}) can be represented as
	$
		\hat{P}_{2,k}^{-1} \succeq \underline{\beta}((1-\overline{\beta})\varphi)^r\underline{\rho}I_{n_2},
	$
	and we can arrive at (\ref{eq:405}).

	\bibliographystyle{ieeetr}

	\bibliography{biblio}

\begin{thebibliography}{10}

\bibitem{c34}
D.~Ding, Q.-L. Han, X.~Ge, and J.~Wang, ``Secure state estimation and control
  of cyber-physical systems: A survey,'' 2021.

\bibitem{c35}
Z.~Zhou, M.~Zhong, and Y.~Wang, ``Fault diagnosis observer and fault-tolerant
  control design for unmanned surface vehicles in network environments,'' {\em
  IEEE Access}, vol.~7, pp.~173694--173702, 2019.

\bibitem{c38}
B.~Chen and G.~Hu, ``Nonlinear state estimation under bounded noises,'' {\em
  Automatica}, vol.~98, pp.~159--168, 2018.

\bibitem{c21}
Y.~Liu, Y.~Zhao, and F.~Wu, ``Ellipsoidal state-bounding-based set-membership
  estimation for linear system with unknown-but-bounded disturbances,'' {\em
  IET Control Theory \& Applications}, vol.~10, no.~4, pp.~431--442, 2016.

\bibitem{c3}
B.~S. Rego, G.~V. Raffo, J.~K. Scott, and D.~M. Raimondo, ``Guaranteed methods
  based on constrained zonotopes for set-valued state estimation of nonlinear
  discrete-time systems,'' {\em Automatica}, vol.~111, p.~108614, 2020.

\bibitem{c52}
D.~Qu, Z.~Huang, Y.~Zhao, G.~Song, K.~Yi, and X.~Zhao, ``Nonlinear state
  estimation by extended parallelotope set-membership filter,'' {\em ISA
  Transactions}, 2021.

\bibitem{c5}
Z.~Wang, C.-C. Lim, and Y.~Shen, ``Interval observer design for uncertain
  discrete-time linear systems,'' {\em Systems \&\ Control Letters}, vol.~116,
  pp.~41--46, 2018.

\bibitem{c14}
{Haotian Zhang}, R.~{Ayoub}, and S.~{Sundaram}, ``State estimation for linear
  systems with unknown inputs: Unknown input norm-observers and bibobs
  stability,'' in {\em 2015 American Control Conference (ACC)}, pp.~4186--4191,
  2015.

\bibitem{c31}
S.~Sundaram, ``Fault-tolerant and secure control systems,'' {\em University of
  Waterloo, Lecture Notes}, 2012.

\bibitem{c11}
F.~{Xu}, S.~{Yang}, and X.~{Wang}, ``A novel set-theoretic interval observer
  for discrete linear time-invariant systems,'' {\em IEEE Transactions on
  Automatic Control}, pp.~1--1, 2020.

\bibitem{c53}
D.~Efimov, L.~Fridman, T.~Raïssi, A.~Zolghadri, and R.~Seydou, ``Interval
  estimation for lpv systems applying high order sliding mode techniques,''
  {\em Automatica}, vol.~48, no.~9, pp.~2365--2371, 2012.

\bibitem{c19}
E.~Scholte and M.~E. Campbell, ``A nonlinear set-membership filter for on-line
  applications,'' {\em International Journal of Robust and Nonlinear Control},
  vol.~13, no.~15, pp.~1337--1358, 2003.

\bibitem{c20}
Y.~Becis-Aubry, M.~Boutayeb, and M.~Darouach, ``State estimation in the
  presence of bounded disturbances,'' {\em Automatica}, vol.~44, no.~7,
  pp.~1867--1873, 2008.

\bibitem{c15}
F.~J. Bejarano, L.~Fridman, and A.~Poznyak, ``Unknown input and state
  estimation for unobservable systems,'' {\em SIAM Journal on Control and
  Optimization}, vol.~48, no.~2, pp.~1155--1178, 2009.

\bibitem{c16}
B.~{Molinari}, ``A strong controllability and observability in linear
  multivariable control,'' {\em IEEE Transactions on Automatic Control},
  vol.~21, no.~5, pp.~761--764, 1976.

\bibitem{c12}
S.~{Sundaram} and C.~N. {Hadjicostis}, ``Delayed observers for linear systems
  with unknown inputs,'' {\em IEEE Transactions on Automatic Control}, vol.~52,
  no.~2, pp.~334--339, 2007.

\bibitem{c30}
H.~K. Khalil and L.~Praly, ``High-gain observers in nonlinear feedback
  control,'' {\em International Journal of Robust and Nonlinear Control},
  vol.~24, no.~6, pp.~993--1015, 2014.

\bibitem{c32}
Y.~Liu, Y.~Zhao, and F.~Wu, ``Ellipsoidal state-bounding-based set-membership
  estimation for linear system with unknown-but-bounded disturbances,'' {\em
  IET Control Theory \& Applications}, vol.~10, no.~4, pp.~431--442, 2016.

\bibitem{c9}
C.~Durieu, E.~Walter, and B.~Polyak, ``Multi-input multi-output ellipsoidal
  state bounding,'' {\em J. Optim. Theory Appl.}, vol.~111, p.~273–303, Nov.
  2001.

\bibitem{c59}
M.~Abate and S.~Coogan, ``Computing robustly forward invariant sets for
  mixed-monotone systems,'' in {\em 2020 59th IEEE Conference on Decision and
  Control (CDC)}, pp.~4553--4559, IEEE, 2020.

\bibitem{c57}
W.~Li, G.~Wei, D.~Ding, Y.~Liu, and F.~E. Alsaadi, ``A new look at boundedness
  of error covariance of kalman filtering,'' {\em IEEE Transactions on Systems,
  Man, and Cybernetics: Systems}, vol.~48, no.~2, pp.~309--314, 2018.

\bibitem{c47}
C.~Edwards and S.~Spurgeon, {\em Sliding Mode Control: Theory And
  Applications}.
\newblock Series in Systems and Control, CRC Press, 1998.

\bibitem{c48}
B.~Zhou, J.~Han, and G.~Liu, ``A ud factorization-based nonlinear adaptive
  set-membership filter for ellipsoidal estimation,'' {\em International
  Journal of Robust and Nonlinear Control: IFAC-Affiliated Journal}, vol.~18,
  no.~16, pp.~1513--1531, 2008.

\bibitem{c36}
L.~K. Vasiljevic and H.~K. Khalil, ``Differentiation with high-gain observers
  the presence of measurement noise,'' in {\em Proceedings of the 45th IEEE
  Conference on Decision and Control}, pp.~4717--4722, 2006.

\end{thebibliography}
	
\end{document}